\documentclass[12pt, draftclsnofoot, onecolumn]{IEEEtran}

\usepackage{cite}
\usepackage[cmex10]{amsmath}
\usepackage{amssymb}
\usepackage{amsthm}
\usepackage{mathrsfs}
\usepackage{bm}
\usepackage[mathscr]{eucal}
\usepackage{amssymb,amsmath,amsthm,amsfonts,latexsym}
\usepackage{amsmath,graphicx,bm,xcolor,url}
\usepackage{graphicx}
\usepackage{latexsym}
\usepackage{CJK}
\usepackage{indentfirst}
\usepackage{geometry}
\usepackage{psfrag}
\usepackage{setspace}
\usepackage{algorithmic}
\usepackage{algorithmic, cite}
\usepackage{algorithm}
\usepackage{array}
\usepackage{mdwmath}
\usepackage{mdwtab}
\usepackage{eqparbox}
\usepackage{url}
\usepackage{epstopdf}
\usepackage{epsfig,epsf,psfrag}
\usepackage{fixltx2e}
\usepackage{verbatim}
\usepackage{textcomp}
\usepackage{psfrag} 

\usepackage{subfigure} 
\usepackage{caption}

\theoremstyle{plain}
\newtheorem{mythe}{Theorem}
\theoremstyle{remark}

\theoremstyle{plain}

\theoremstyle{remark}

\theoremstyle{plain}

\theoremstyle{remark}

\theoremstyle{remark}

\theoremstyle{remark}

\theoremstyle{remark}

\theoremstyle{remark}

\theoremstyle{remark}

\catcode`~=11 \def\UrlSpecials{\do\~{\kern -.15em\lower .7ex\hbox{~}\kern .04em}} \catcode`~=13

\allowdisplaybreaks[4]

\newcommand{\calB}{\mathcal{B}}

\newcommand{\calK}{\mathcal{K}}

\newcommand{\calM}{\mathcal{M}}

\newcommand{\bB}{\mathbf{B}}

\newcommand{\bp}{\mathbf{p}}
\newcommand{\bP}{\mathbf{P}}
\newcommand{\bq}{\mathbf{q}}
\newcommand{\bQ}{\mathbf{Q}}

\newcommand{\bt}{\mathbf{t}}

\newcommand{\bu}{\mathbf{u}}

\newcommand{\bw}{\mathbf{w}}

\newcommand{\bx}{\mathbf{x}}
\newcommand{\bX}{\mathbf{X}}
\newcommand{\by}{\mathbf{y}}

\newcommand{\bz}{\mathbf{z}}

\newcommand{\rmd}{\mathrm{d}}

\newcommand{\bbR}{\mathbb{R}}

\DeclareMathAlphabet{\mathbsf}{OT1}{cmss}{bx}{n}
\DeclareMathAlphabet{\mathssf}{OT1}{cmss}{m}{sl}

\DeclareSymbolFont{bsfletters}{OT1}{cmss}{bx}{n}
\DeclareSymbolFont{ssfletters}{OT1}{cmss}{m}{n}
\DeclareMathSymbol{\bsfGamma}{0}{bsfletters}{'000}
\DeclareMathSymbol{\ssfGamma}{0}{ssfletters}{'000}
\DeclareMathSymbol{\bsfDelta}{0}{bsfletters}{'001}
\DeclareMathSymbol{\ssfDelta}{0}{ssfletters}{'001}
\DeclareMathSymbol{\bsfTheta}{0}{bsfletters}{'002}
\DeclareMathSymbol{\ssfTheta}{0}{ssfletters}{'002}
\DeclareMathSymbol{\bsfLambda}{0}{bsfletters}{'003}
\DeclareMathSymbol{\ssfLambda}{0}{ssfletters}{'003}
\DeclareMathSymbol{\bsfXi}{0}{bsfletters}{'004}
\DeclareMathSymbol{\ssfXi}{0}{ssfletters}{'004}
\DeclareMathSymbol{\bsfPi}{0}{bsfletters}{'005}
\DeclareMathSymbol{\ssfPi}{0}{ssfletters}{'005}
\DeclareMathSymbol{\bsfSigma}{0}{bsfletters}{'006}
\DeclareMathSymbol{\ssfSigma}{0}{ssfletters}{'006}
\DeclareMathSymbol{\bsfUpsilon}{0}{bsfletters}{'007}
\DeclareMathSymbol{\ssfUpsilon}{0}{ssfletters}{'007}
\DeclareMathSymbol{\bsfPhi}{0}{bsfletters}{'010}
\DeclareMathSymbol{\ssfPhi}{0}{ssfletters}{'010}
\DeclareMathSymbol{\bsfPsi}{0}{bsfletters}{'011}
\DeclareMathSymbol{\ssfPsi}{0}{ssfletters}{'011}
\DeclareMathSymbol{\bsfOmega}{0}{bsfletters}{'012}
\DeclareMathSymbol{\ssfOmega}{0}{ssfletters}{'012}

\newcommand{\tild}{\widetilde{d}}

\newcommand{\barE}{\bar{E}}

\newcommand{\barK}{\bar{K}}

\newcommand{\barQ}{\bar{Q}}

\def\norm#1{\left\| #1 \right\|}
\def\norm2#1{\left\| #1 \right\|_2}
\def\norm22#1{\left\| #1 \right\|_2^2}

\newcommand{\eqa}{\stackrel{(a)}{=}}

\newcommand{\leb}{\stackrel{(b)}{\le}}
\newcommand{\lec}{\stackrel{(c)}{\le}}

\newcommand{\qednew}{\nobreak \ifvmode \relax \else
      \ifdim\lastskip<1.5em \hskip-\lastskip
      \hskip1.5em plus0em minus0.5em \fi \nobreak
      \vrule height0.75em width0.5em depth0.25em\fi}

\usepackage{float}
\usepackage{cite}
\usepackage{times,amsmath,color,amssymb,graphicx,epsfig,multirow,float,algorithm,algorithmic}

\title{Energy-Efficient UAV Backscatter Communication with Joint Trajectory Design and Resource Optimization}
\author{Gang~Yang, \emph{Member, IEEE}, Rao Dai, \emph{Student Member, IEEE}, and Ying-Chang~Liang, \emph{Fellow, IEEE}
\thanks{The conference version \cite{UAVBackComYangDaiICC2019} of this paper was presented in IEEE International Conference on Communications, Shanghai, China, in May 2019. }
 \thanks{G.~Yang and R. Dai are with the National Key Laboratory of Science and Technology on Communications, and Center for Intelligent Networking and Communications (CINC), University of Electronic Science and Technology of China (UESTC), Chengdu 611731, China (e-mails: yanggang@uestc.edu.cn, dairao@std.uestc.edu.cn).}
\thanks{Y.-C. Liang is with Center for Intelligent Networking and Communications (CINC), University of Electronic Science and Technology of China (UESTC), Chengdu 611731, China (e-mail: liangyc@ieee.org). (\emph{Corresponding author: Y.-C. Liang.})}}

\begin{document}
 \maketitle

\begin{abstract}
Backscatter communication which enables wireless-powered backscatter devices (BDs) to transmit information by reflecting incident signals, is an energy- and cost-efficient communication technology for Internet-of-Things. This paper considers an unmanned aerial vehicle (UAV)-assisted backscatter communication network (UBCN) consisting of multiple BDs and carrier emitters (CEs) on the ground as well as a UAV. A communicate-while-fly scheme is first designed, in which the BDs illuminated by their associated CEs transmit information to the flying UAV in a time-division-multiple-access manner. Considering the critical issue of the UAV's limited on-board energy and the CEs' transmission energy, we maximize the energy efficiency (EE) of the UBCN by jointly optimizing the UAV's trajectory, the BDs' scheduling, and the CEs' transmission power, subject to the BDs' throughput constraints and harvested energy constraints, as well as other practical constraints. Furthermore, we propose an iterative algorithm based on the block coordinated decent method to solve the formulated mixed-integer non-convex problem, in each iteration of which the variables are alternatively optimized by leveraging the cutting-plane technique, the Dinkelbach's method and the successive convex approximation technique. Also, the convergence and complexity of the proposed algorithm are analyzed. Finally, simulation results show that the proposed communicate-while-fly scheme achieves significant EE gains compared with the benchmark hover-and-fly scheme. Useful insights on the optimal trajectory design and resource allocation are also obtained.
\end{abstract}


\section{Introduction}
 \subsection{Motivation}
Internet-of-Things (IoT) is revolutionizing the way we live by providing ubiquitous connectivity among billions of devices \cite{ZhangLiangXiaoWC18}. Backscatter communication (BackCom) enables passive backscatter devices (BDs) to transmit information by modulating incident sinusoidal carriers or ambient radio-frequency (RF) carriers without using any power-hungry or complex RF transmitters, and thus is an energy- and cost-efficient communication technology for IoT devices that typically have limited battery energy and strict cost constraint~\cite{BletsasSPM18,BoyerRoy,XuYangZhang,ABCSigcom13,YangLiangZhangPeiTCOM17,BletsasKimionis14,BletsasAlevizos15,LoreaVarshney2017,GongNiyato18,YangXuLiangWCL19}. Specifically, the bistatic BackCom (BBC) systems with spatially separated carrier emitters (CEs) and backscatter receivers (BRs) \cite{BletsasKimionis14} were demonstrated to achieve a communication range on the order of hundreds of meters \cite{BletsasAlevizos15} \cite{LoreaVarshney2017}, and has various applications such as monitoring environmental humidity and soil moisture~\cite{DaskalakisBletsasTMTT16}.
However, the current BBC systems with fixed BRs face two main challenges. First, it is cost-inefficient to directly deploy a BBC network for data collection in large-scale IoT. Since the communication range of BBC is shorter than that of traditional communication with active radio, many expensive BRs are needed to cover massive BDs. Second, the transmission rates of distributed BDs suffer from fairness issue. For a BD located far away from both its associated CE and the BR, the backscattered signals endure severe channel fading twice, resulting into a low data rate compared with that of other nearby BDs.


Unmanned aerial vehicle (UAV)-assisted wireless communication has attracted growing research interests from both academy and industry, due to its advantages in flexible deployment, fully controllable mobility, and high probability of line-of-sight (LoS) links from air to ground~\cite{ZengZhangCM16}. UAVs have a great potential in enhancing the performance of wireless communications, such as providing assistance for existing terrestrial cellular networks~\cite{SharmaBennisKumar16,ChengZhangOffloading,ChenZhaoRelays}, and enabling information dissemination and collection in wireless sensor networks~\cite{MozaffariDebbahTWC16,WuZhangTWC18,XuZhangIoTJ18}.

Motivated by the superiority in UAV-assisted wireless communications, we introduce a UAV to act as a flying BR for a BBC network. The UAV can adjust its flying trajectory to obtain higher transmission rate by exploiting better air-to-ground channel conditions, and achieve better BD fairness by intelligently scheduling the BDs to transmit to the flying UAV. The UAV not only has lower cost than expensive BRs, but also is flexible and efficient to collect data from massive BDs in a large-scale BackCom network. Since energy is critical for passive IoT networks and also vitally important for UAV due to its limited on-board energy, we aim to maximize the energy efficiency (EE) of  such a UAV-assisted  BackCom network (UBCN) in this paper.

\subsection{Related Works}
The existing BackCom systems can be divided into four categories, namely the monostatic BackCom (MBC) system \cite{Dobkinbook2007,YangBackscatter15,GuoZhouYanikomeroglu18,MishraLarsson19} (e.g., radio-frequency identification) with co-located illuminating CEs and BRs, the BBC system~\cite{BletsasKimionis14,BletsasAlevizos15,LoreaVarshney2017,GongNiyato18,YangXuLiangWCL19}, the ambient BackCom (ABC) system~\cite{ABCSigcom13,YangLiangZhangPeiTCOM17,QianGaoAmBCTWC16,ZhangLiangJSAC19} with separately located ambient transmitters (e.g., WiFi access point, cellular base station) and BRs, as well as the symbiotic radio (SR) system with co-located legacy transmitter/receiver and BR \cite{YangLiangZhangIoTJ18, LongGuoYangArxiv18,ZhangLiangSRAccess19, LiangLISA2019,LongLiangSRAccess19}. The SR system enables the backscatter transmission to share both the spectrum and the infrastructure of existing legacy wireless communication systems. The communication coverage of traditional MBC is inherently limited to tens of meters by the short distance of energy transfer to passive BDs like tags, due to fast decaying of electromagnetic waves with respect to distance \cite{Dobkinbook2007}. The ABC also has a limited coverage owing to the strong direct-link interference from uncontrollable ambient transmitters, which increases the difficulty in decoding BD's information. In contrast, the BBC enables BDs to communicate with remote BRs hundreds of meters away, since the CE-to-BD distance can be short and the direct-link carrier signals can be easily removed without introducing interference~\cite{BletsasAlevizos15}.

UAVs can exploit the high mobility to change their positions and fly near to low-power IoT devices for strong air-to-ground LoS links, which makes them especially suitable for information dissemination and collection in IoT networks \cite{ZengZhangCM16}.
In \cite{WuZhangTWC18}, the minimal average rate of ground users was maximized
based on an optimal joint design of the UAV's trajectory, the power control and the uplink access scheduling. In \cite{MozaffariDebbahTWC17}, the optimal deployment and mobility of multiple UAVs were investigated to minimize the transmission power of ground IoT devices. The EE of a fixed-wing UAV-assisted wireless communication system was studied in \cite{ZengZhangEETWC17}, but this work is limited to the case of a single UAV and a single ground user. In \cite{Zeng2018Energy}, the minimal energy consumption of a UAV serving multiple IoT devices was studied based on a theoretical model of the propulsion energy consumption for rotary-wing UAVs. In \cite{HuangYangCL18}, the power consumption of the UAV was minimized while guaranteing the required transmission rate of multiple sensor nodes.

Recently, the authors of \cite{FarajzadehYanikomeroglu} studied a UAV-assisted MBC network, in which the UAV acts both as a flying CE and as a flying BR. The BDs adopt a non-orthogonal-multiple-access (NOMA) scheme to transmit information to the BR. The optimal UAV flying altitude was optimized to maximize the number of successfully decoded bits in the uplink while minimizing the flight time. Our paper differs from this work \cite{FarajzadehYanikomeroglu} in the following two main aspects. First, our paper has a different system model and a different multiple-access scheme, i.e., we study the UAV-assisted BBC network with a time-division-multiple-access (TDMA) transmission scheme. Second, our paper has different design objective and optimization dimension, i.e., we aim to design an energy-efficient UBCN via jointly optimizing UAV's trajectory and system resources.

\subsection{Contributions}
As shown in Fig.~\ref{fig:systemdes}, this paper investigates a UBCN, where a flying rotary-wing UAV collects data from multiple BDs on the ground. We aim to maximize the EE of such a UBCN. The main contributions are summarized as follows: 
\begin{itemize}
  \item
  First, a communicate-while-fly scheme is proposed for the UBCN. The BDs harvest energy from the incident sinusoidal signals emitted by their associated CEs, and transmit information to the flying UAV in a TDMA manner. We derive the BDs' throughput performance depending on the BDs' scheduling, the CEs' transmission power as well as the UAV's flying trajectory, and analyze the system's total energy consumption consisting of the CEs' transmission energy and the UAV's energy consumption.
\item Second, we formulate an optimization problem to maximize the EE of the UBCN by jointly optimizing three blocks of variables including the BDs' scheduling, the CEs' transmission power and the UAV's flying trajectory, subject to each BD's minimal throughput constraint and minimal harvested energy constraint, the UAV's maximum speed constraint and other practical constraints. The formulated problem is appealing, since the EE performance can benefit from multiple design dimensions. However, the problem is non-convex and challenging to be solved optimally.
  \item Third, we propose an efficient iterative algorithm to solve the EE maximization problem. A time-discretization approach is applied to transform the original problem with integrals of continuous-time variables into a problem with discrete-time variables. Then, an iterative algorithm based on the block-coordinated-decent (BCD) method is proposed to solve the time-discrete problem. Specifically, in each iteration, each block of variables are optimized alternately with the other two blocks of variables fixed, by leveraging the cutting-plane technique, the Dinkelbach's method and the successive convex approximation (SCA) technique. In addition, the complexity and convergence of the proposed algorithm are analyzed.
  \item Finally, numerical results show that significant EE gains can be achieved by the proposed communicate-while-fly scheme, compared with the benchmark hover-and-fly scheme in which the UAV sequentially visits different positions and hovers for a certain time at each position to receive signals from ground BDs. In order to maximize the EE, the time resource prefers to be allocated to those BDs with stronger reflection power. Also, the UAV trajectory is optimized to balance the BDs' total throughput and the UAV's energy consumption.
\end{itemize}

The rest of this paper is organized as follows. Section~\ref{systemmodel} describes the system model of the UBCN and presents the designed communicate-while-fly scheme. Section~\ref{formulation} formulates the EE maximization problem with joint optimization of the UAV's trajectory and system resources. Section~\ref{solution} presents an efficient iterative algorithm to solve the EE optimization problem. Section~\ref{hover-and-fly} introduces the hover-and-fly scheme as a practical benchmark. Section~\ref{simulation} shows numerical results to verify the performance of the proposed communicate-while-fly scheme and the designed algorithm. Finally, Section~\ref{conslusion} concludes this paper. 

We use the following notations throughout this paper. Bold uppercase letters $\bX$, bold lowercase letters $\bx$ and nonbold lowercase letters $x$ denote matrices, column vectors and scalers, respectively. $\bX^{\text{T}}$ and $\|\bx\|$ denote the transpose and the Euclidean norm, respectively; $\dot{x}(t)$ is the first derivative of function $x(t)$; $|\mathcal{X}|$ represents the cardinality of set $\mathcal{X}$; and $\bbR^{x \times y}$ denotes the space of $x \times y$ real matrices. Finally, $\triangleq$ means the equivalence in definition.

\section{System Model}\label{systemmodel}
\subsection{System Description}
As shown in Fig.~\ref{fig:systemdes}, we consider a UBCN in which a flying rotary-wing UAV is deployed to collect data from $K \ (K \geq 1)$ passive BDs illuminated by $M \ (M \geq 1)$ CEs on the ground. A BBC network structure is adopted for the UBCN, where passive BDs are uniformly distributed and multiple CEs are deployed for seamless coverage of this area, with the sets of $\mathcal{K} \triangleq \{1, \ldots, K\}$ and $\mathcal{M} \triangleq \{1, \ldots, M\}$, respectively.
We assume that each BD is associated to its geographically closest CE. Without loss of generality, the BDs associated to CE $m \in \calM$ are denoted by the set $\calK_m$ with cardinality $K_m \triangleq |\calK_m|$. Clearly, $\calK=\calK_1 \bigcup \calK_2 \bigcup \cdots \bigcup \calK_M$. The maximum number of BDs supported by a CE is defined as $\barK \triangleq \max \{K_1, \ldots, K_M\}$. In order to transmit information to the UAV, BD $k \in \calK_m$ modulates the incident RF signals that are transmitted from its associated CE $m$ by intentionally switching its load impedance to change the amplitude and/or phase of the reflected signals. BD $k$ also splits a portion of the incident RF signals for energy harvesting to support its circuit operation. The UAV is assumed to fly at a fixed altitude $H \ (H>0)$ within $T \  (T>0)$ seconds (s) to collect data. In practice, the UAV altitude $H$ needs to be properly chosen according to the network coverage and the communication range constraint of BBC \cite{BletsasKimionis14}. In addition, all the CEs send carriers at the same frequency $f_c$.~

\begin{figure} [t!]
	\centering \includegraphics[width=0.7\columnwidth]{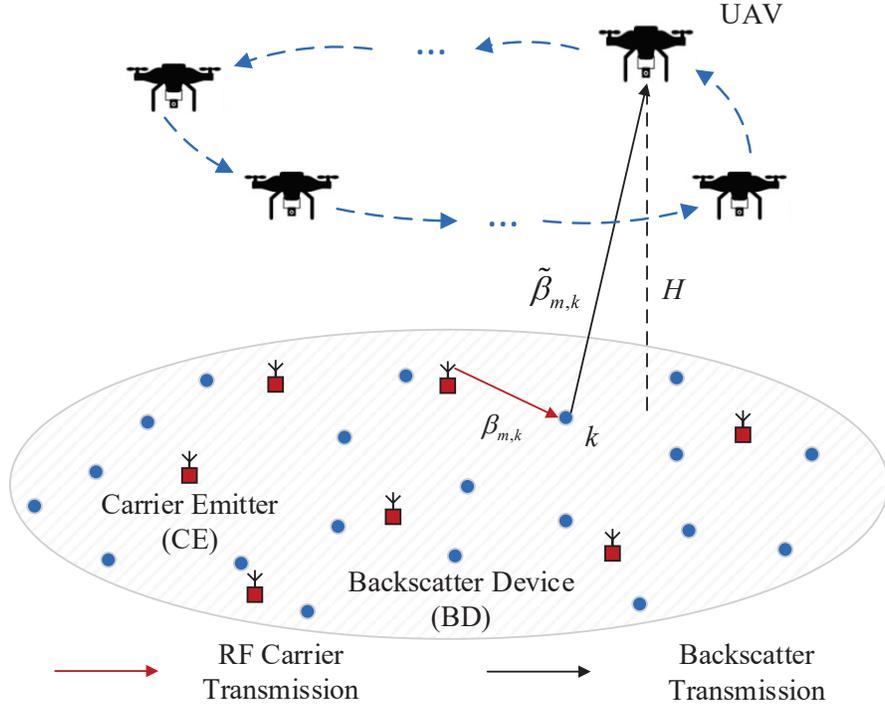}
\caption{System Description.} \label{fig:systemdes}
\end{figure}



For convenience of analysis, a two-dimensional Cartesian coordinate system is considered herein. The locations of CE $m \in \calM$ and BD $k \in \calK_m$ are denoted by $\mathbf{u}_{m} \in \bbR^{2 \times 1}$ and $\mathbf{w}_{m,k} \in \bbR^{2 \times 1}$, respectively. The exact locations of the CEs and BDs on the ground are assumed to be known by a central controller of the system. The UAV's position projected onto the horizontal plane at time $t \in [0,T]$ is denoted by $\mathbf{q}(t)=[x(t), y(t)]^{\text{T}} \in \bbR^{2 \times 1}$. The fixed distance between BD $k$ and CE $m$ is given by $d_{m,k}= \|\bw_{m,k}-\bu_{m}\|$, and the time-varying distance between BD $k$ and the UAV is denoted by $\tild_k (t)=\sqrt{H^2+\|\bw_{m,k}-\bq(t)\|^2}$.
Moreover, we assume that the CE-to-BD links are Rayleigh channels with distance-dependent large-scale fading coefficients, and the BD-to-UAV links are dominated by the LoS links, which differs from the channel links of terrestrial communication systems \cite{ZengZhangCM16, Zeng2018Energy}.
Based on the free-space path loss model, the channel power gain between CE $m$ and BD $k \in \calK_m$ can be expressed as $\beta_{m,k}=\frac{\beta_{0}}{\left\|\bw_{m,k}-\bu_{m}\right\|^2}$, where $\beta_{0}$ is the channel power gain at a reference distance of one meter (m). Similarly, the channel power gain between BD $k$ and the UAV is given by ${\tilde{\beta}}_{m,k}(t) =\frac{\beta_{0}}{H^2+\left\|\bw_{m,k}-\bq(t)\right\|^2}$.



\subsection{Communicate-while-fly Scheme}
We design a communicate-while-fly scheme for the UBCN. That is, the BDs illuminated by their associated CEs transmit signals to the UAV in a TDMA manner, and the UAV receives signals while flying along its designed trajectory.

For each BD $k \in \calK_m$, we use the binary variable $b_{m,k} (t) \in \{0, 1\}$ to represent its scheduling indicator at time $t$, with $b_{m,k}(t) =1$ indicating that it is scheduled to communicate with the UAV, otherwise $b_{m,k}(t) =0$. At time $t$, at most one BD is allowed to transmit signals in the uplink, and thus we have the scheduling constraint $\sum_{m=1}^{M} \sum_{k=1}^{\barK} b_{m,k}(t) \leq 1, \;\; \forall t \in [0, T]$. Denote all BDs' scheduling indicators at time $t$ by matrix $\bB(t) \in \bbR^{M \times \barK}$.

Depending on the reflecting status of each BD, a proportion of the BD's incident power is reflected, and the remained power is used for energy harvesting. In practice, BDs have finite reflecting status which are once implemented and then fixed. In this paper, following the parameters of most commercial tags \cite{TagDataSheetImpinpj}, we assume that each BD either fully backscatters its incident signals or uses all the incident signals for energy harvesting\footnote{In practice, typical commercial RFID tags like Impinj Monza R6 \cite{TagDataSheetImpinpj} adopt the two-state amplitude-shift-keying modulation.}.

Denote the transmission power of CE $m$ by $P_{m}(t) \in [0, P_{\text{max} }]$, and the corresponding vector by $\bp(t) = [ P_1(t), \ldots, P_M(t)]^T \in \bbR^{M \times 1 }$. Let $\eta_{m,k} \in [0, 1]$ be the RF-energy harvesting efficiency of each BD. Thus the energy harvested by BD $k \in \calK_m$ can be expressed as 
\begin{align}
  E_{m, k}= \int_{0}^{T}\eta_{m,k} \beta_{m,k}P_{m}(t)\left(1-b_{m,k}(t)\right)\rmd t.
\end{align}

\subsection{Energy Efficiency Performance}
The signal-to-noise-ratio (SNR) of signals received by the UAV from BD $k \in \calK_{m}$ can be expressed as $\gamma_{m,k} = P_{m}(t) b_{m,k}(t){\tilde{\beta}}_{m,k}(t)  \beta_{m,k}/\sigma^{2}$, where $\sigma^{2}$ is the additive white Gaussian noise power at the UAV receiver. The throughput of BD $k$ normalized to bandwidth during the period of $T$ in bits-per-Hertz (bits/Hz), denoted by $Q_{m,k}$, is given by
\begin{equation}\label{eq:throughput_k}
   Q_{m,k} (\bB(t), \bp(t), \bq(t)) =
        \int_{0}^{T} b_{m, k}(t)  \log_{2}\left(1+\frac{{\beta_0}  \beta_{m,k} P_{m}(t)}{\sigma^{2} \left(H^2+\left\|\bw_{m,k}-\bq(t)\right\|^2\right)}\right) \rmd t.
\end{equation}



The UBCN is energy-constrained due to the UAV's limited on-board battery, which is a crucial issue that should be carefully considered.
In general, the UAV's energy consumption consists of the propulsion energy and the communication energy.
We ignore the communication energy, as it is typically much smaller than the propulsion energy\cite{ZengZhangEETWC17}\cite{Zeng2018Energy}.
Moreover, we adopt the following propulsion-power model of rotary-wing UAVs \cite{Zeng2018Energy}
\begin{equation}\label{PV}
  P_{\text{UAV}}(t)=P_{\text{b}}\negthinspace \left(1+\frac{3V(t)^2}{U_{\text{tip}}^2}\right)  +
  P_{\text{i}}\negthinspace\left(\sqrt{1+\frac{V(t)^4}{4v_{0}^4}}-\frac{V(t)^2}{2v_{0}^2}\right)^{\mspace{-10mu}\frac{1}{2}}
   +\frac{1}{2}d_0\theta sAV(t)^3,
\end{equation}
where $V(t)=\sqrt{|\dot{x}(t)|^2+|\dot{y}(t)|^2}$ is the UAV's flying speed in meters-per-second (m/s),
$P_{\text{b}}$ and $P_{\text{i}}$ represent the blade profile power and induced power in hovering status (i.e., $V(t)=0$) respectively, $U_{\text{tip}}$ is the tip speed of the rotor blade, $v_{0}$ is the mean rotor induced velocity, and $\theta$ is the air density. Other parameters in \eqref{PV} depend on the UAV's properties and environment conditions~\cite{Zeng2018Energy}.

Considering the energy consumption of both the UAV and the CEs, from \eqref{eq:throughput_k} and \eqref{PV}, the EE of the UBCN is given by 
\begin{align}
  \mathbf{EE}=\frac{\sum\limits_{m=1}^{M}\!\sum\limits_{k=1}^{\barK} \int_{0}^{T} b_{m,k}(t)\log_{2}\left(  1 + \frac{\beta_{0}\,\beta_{m,k}P_{m}(t)}{\sigma^{2}\left(\! H^2+\left\|\bw_{m,k}-\bq(t)\right\|^2\!\right)}\right) \rmd t}{\int_{0}^{T}\left(P_{\text{UAV}}(t)+\sum\limits_{m=1}^{M}P_{m}(t)\right)\rmd t}.
\end{align}

It is noted that the above EE for the UBCN differs from that for a single ground user sending data to a flying fixed-wing UAV via active transmission \cite{ZengZhangEETWC17}. Although the reference \cite{XuZhangIoTJ18} optimized the total throughput of multiple ground devices which are sequentially powered by and transmit to a flying UAV, the EE of that whole network is not studied.

\section{Problem Formulation}\label{formulation}
In this section, we formulate the EE optimization problem for such a UBCN, and give the transferred tractable problem via the time-discretization method. The objective is to maximize the EE of the UBCN by jointly optimizing the BDs' scheduling $\bB(t)$, the CEs' transmission power $\bp(t)$ and the UAV's flying trajectory $\bq(t)$. Mathematically, the optimization problem can be formulated as 
\begin{subequations}\label{eq:PCT}
\begin{align}
\underset{\begin{subarray}{c}\bB(t),\bp(t),\\\bq(t)\end{subarray}}{\max}&\quad
\frac{\sum\limits_{m=1}^{M}\sum\limits_{k=1}^{\barK} \int_{0}^{T} b_{m,k}(t)\log_{2}\left( 1+\frac{\beta_{0}\,\beta_{m,k}P_{m}(t)}{\sigma^{2}\left( H^2+\left\|\bw_{m,k}-\bq(t)\right\|^2\right)}\right) \rmd t}{\int_{0}^{T}\left(P_{\text{UAV}}(t)+\sum\limits_{m=1}^{M}P_{m}(t)\right)\rmd t}\label{eq:ObjPCT}\\
\text{s.t.} &\quad  \int_{0}^{T} b_{m,k}(t)\log_{2}\left(1+\frac{\beta_{0}\,\beta_{m,k}P_{m}(t)}{\sigma^{2}\left(H^2+\left\|\bw_{m,k}-\bq(t)\right\|^2\right)}
\right)\rmd t \geq \barQ_{m,k},\quad\forall m,k &\label{eq:C1PCT}\\
&\quad \int_{0}^{T}\eta_{m,k}\beta_{m,k}P_{m}(t)\left(1-b_{m,k}(t)\right)\rmd t \geq \barE_{m,k},\forall m,k \label{eq:C2PCT}\\
&\quad \sum\limits_{m=1}^{M}\sum\limits_{k=1}^{\barK}b_{m,k}(t) \leq 1,\quad\forall t \label{eq:C3PCT}\\
&\quad b_{m,k}(t) \in \{0,1\},\quad\forall m,k,t \label{eq:C4PCT}\\
&\quad 0\leq P_{m}(t)\leq P_{\text{max}},\quad\forall m,t \label{eq:C5PCT}\\
&\quad |\dot{x}(t)|^2+|\dot{y}(t)|^{2}\leq V_{\text{max}}^2,\quad\forall t \label{eq:C7PCT}\\
&\quad \bq(0)=\bq(T) \label{eq:C8PCT},
\end{align}
\end{subequations}
where \eqref{eq:C1PCT} is the required minimum throughput $\barQ_{m,k}$ of each BD $k \in \calK_m$, \eqref{eq:C2PCT} is the required minimum harvested energy $\barE_{m,k}$ of each BD $k \in \calK_m$, \eqref{eq:C3PCT} is the TDMA scheduling constraint for all BDs, \eqref{eq:C4PCT} is the binary constraint for the scheduling indicator variables, \eqref{eq:C5PCT} is the transmission power constraint of each CE, \eqref{eq:C7PCT} is the flying speed constraint of the UAV, and \eqref{eq:C8PCT} is the practical constraint that the UAV flies back to its start point after the operation time period of $T$ s.

The above joint optimization problem is appealing in practice. For one thing, the UAV trajectory design and the BDs' scheduling can be exploited to satisfy the uplink throughput requirement; for another, the EE can be further improved through properly allocating the CEs' transmission power and optimizing the UAV's flying speed which is included in the UAV trajectory design.

However, problem \eqref{eq:PCT} cannot be solved directly due to the following two main challenges. First, variables in~\eqref{eq:PCT} are all continuous functions of time $t$ and both the nominator and the denominator of \eqref{eq:ObjPCT} include integral formulas.
Second, the variables $\bB(t)$, $\bp(t)$, and $\bq(t)$ are always coupled with each other in \eqref{eq:ObjPCT}, \eqref{eq:C1PCT}, and \eqref{eq:C2PCT}, which makes the problem especially complicated.
To tackle the first challenge, we apply the time-discretization method \cite{WuZhangTWC18},
 in which the UAV's operation time $T$ is divided into $N$ equal and sufficiently short time slots with index $n = 1,\ldots,N$. The time duration of each slot is $T_s=T/N$.
Therefore, in each time slot, the UAV's position can be considered stationary and the BD-to-UAV channels are supposed to be invariant. To tackle the second challenge, we utilize an alternative optimization technique (i.e., BCD method) to decouple the three blocks of variables, which will be discussed in the next section.

The corresponding discrete-time variables can be denoted by $\calB=\{\bB(1), \ldots, \bB(N)\} \in \bbR^{M \times \barK \times N}$, $\bP=[\bp(1), \ldots, \bp(N)] \in \bbR^{M \times N}$, and $\bQ=[\bq(0), \bq(1), \ldots, \bq(N)] \in \bbR^{2 \times (N+1)}$. Since $T_s$ is small enough, the UAV's flying speed is considered to be constant in each time slot, $V(n)=\|\bq(n)-\bq(n-1)\|/T_{s}$.
Let $P_{\text{UAV}}(n)$ be the UAV's power consumption function of time slot $n$, with time argument $t$ replaced by $n T_s$ in \eqref{PV}. Hence, the equivalent EE optimization problem in discrete-time form can be written as  
\begin{subequations}\label{eq:PDT}
\begin{align}
\underset{\begin{subarray}{c}\calB, \bP,\bQ \end{subarray}}{\max} &\quad
\frac{ \sum\limits_{n=1}^{N}\sum\limits_{m=1}^{M}\sum\limits_{k=1}^{\barK} b_{m,k}(n)\log_{2}\left(
1+\frac{\beta_{0}\,\beta_{m,k}P_{m}(n)}{\sigma^{2}\left(H^2+\left\|\bw_{m,k}-\bq(n)
\right\|^2\right)}\right)}
{\sum\limits_{n=1}^{N}\left[P_{\text{UAV}}(n)+\sum\limits_{m=1}^{M}P_{m}(n)\right]} \label{eq:ObjPDT} \\
\text{s.t.} &\quad  T_{s}
\sum\limits_{n=1}^{N}b_{m,k}(n)\log_{2}  \left(  1 +\frac{\beta_{0}\,\beta_{m,k}P_{m}(n)}{\sigma^{2} \left( {H^2+ \|\bw_{m,k}-\bq(n) \|^2} \right)} \right)\geq \barQ_{m,k},\quad\forall m,k \label{eq:C1PDT} \\
&\quad T_{s}\eta_{m,k}\beta_{m,k}\sum\limits_{n=1}^{N} P_{m}(n)\left(1-b_{m,k}(n)\right)\geq\barE_{m,k}, \forall m,k\label{eq:C2PDT}\\
&\quad\sum\limits_{m=1}^{M}\sum\limits_{k=1}^{\barK}b_{m,k}(n) \leq 1,\quad\forall n\label{eq:C3PDT}\\
&\quad b_{m,k}(n) \in \{0,1\},\quad\forall m,k,n\label{eq:C4PDT}\\
&\quad 0\leq P_{m}(n)\leq P_{\text{max}},\quad\forall m,n\label{eq:C5PDT}\\
&\quad \left\|\bq(n)-\bq(n-1)\right\|\leq V_{\text{max}}T_{s},\quad\forall n\label{eq:C6PDT}\\
&\quad \bq(0)=\bq(T)\label{eq:C7PDT}.
\end{align}
\end{subequations}

The numerator of the objective \eqref{eq:ObjPDT} is an integer-weighted sum of logarithm functions, and the denominator of \eqref{eq:ObjPDT} contains a linear combination of complicated non-convex functions $P_{\text{UAV}}(n)$'s. Furthermore, the left-hand-sides (LHSs) of the constraints \eqref{eq:C1PDT} and \eqref{eq:C2PDT} are all non-convex functions of coupled variables $\calB$, $\bP$, and $\bQ$; and \eqref{eq:C4PDT} is an integer constraint. Hence, the equivalent problem \eqref{eq:PDT} is a non-convex and mixed-integer optimization problem, which is too difficult to obtain a globally optimal solution.

\section{Optimal Solution for EE Maximization}\label{solution}
In this section, we propose an efficient iterative algorithm to solve problem \eqref{eq:PDT} based on the BCD method~\cite{TsengBCD01, HongLuoBCDSPM17}, the cutting-plane technique~\cite{schrijver1998}, the Dinkelbach's method \cite{Crouzeix1991} and the SCA technique~\cite{Beck2010}. In each iteration, the three blocks of variables are optimized alteratively, corresponding to three optimization subproblems respectively. Specifically, for any given CEs' transmission power $\bP$ and the UAV's trajectory $\bQ$, we optimize the BDs' scheduling matrix $\calB$ by solving a linear mixed-integer programming (MIP) with the cutting-plane method; for any given BDs' scheduling matrix $\calB$ and the UAV's trajectory $\bQ$, we optimize the CEs' transmission power $\bP$ by solving a fractional programming (FP) with the Dinkelbach's method;
and for any given BDs' scheduling matrix $\calB$ and the CEs' transmission power $\bP$, we optimize the UAV's trajectory $\bQ$ by jointly utilizing the Dinkelbach's method and the SCA technique. Also, the convergence and complexity of the proposed algorithm are analyzed.


\subsection{Uplink Scheduling Optimization}
In iteration $l \ (l \geq 1$), given the CEs' transmission power $\bP^{\{l\}}$ and the UAV's trajectory $\bQ^{\{l\}}$, the BDs' scheduling matrix $\calB$ can be optimized by solving the following problem
\begin{subequations}\label{eq:PDT.1}
\begin{align}
\underset{\calB}{\max}  &\quad
\frac{\sum\limits_{n=1}^{N}\sum\limits_{m=1}^{M}\sum\limits_{k=1}^{\barK}b_{m,k} (n) \log_{2}\left(1+c_1^{\{l\}}(n)\right)} {\sum \limits_{n=1}^{N}P_{\text{UAV}}^{\{l\}}(n) + \sum \limits_{n=1}^{N}\sum \limits_{m=1}^{M}P_{m}^{\{l\}}(n)}\label{eq:ObjPDT.1}\\
\text{s.t.} &\quad \eqref{eq:C1PDT},\eqref{eq:C2PDT},\eqref{eq:C3PDT},\eqref{eq:C4PDT},
\end{align}
\end{subequations}
where the constant $c_1^{\{l\}}(n)$ is given by
\begin{align}
  c_1^{\{l\}}(n)=\frac{\beta_{0}\,\beta_{m,k}P_{m}^{\{l\}} (n)}{\sigma^{2}\left(H^2+\left\|\bw_{m,k}-\bq^{\{l\}} (n) \right\|^2\right)}.\nonumber
\end{align}

%

Problem \eqref{eq:PDT.1} is a linear MIP problem, due to the binary constraint \eqref{eq:C4PDT}. MIP problems can be solved by several existing techniques, such as the branch-and-bound procedure, cutting-plane technique and group-theoretic technique~\cite{schrijver1998}. The cutting-plane technique is adopted in this paper due to its low complexity.

\vspace{-0.1cm}
\subsection{Transmission Power Optimization}
Given the BDs' scheduling matrix $\calB^{\{l\}}$ and the UAV's trajectory $\bQ^{\{l\}}$, the CEs' transmission power $\bP$ can be optimized by solving the following problem
\begin{subequations}\label{eq:PDT.3}
\begin{align}
\underset{\bP}{\max}  &\quad
\frac{  T_{s}\sum \limits_{n=1}^{N} \sum \limits_{m=1}^{M} \sum \limits_{k=1}^{\barK}b_{m,k}^{\{l\}} (n)\log_{2}\left(1+ c_{2}^{\{l\}}(n) P_{m}(n)\right)}{E_{\text{UAV}}^{\{l\}} + T_{s}\sum\limits_{n=1}^{N} \sum\limits_{m=1}^{M}P_{m}(n)}\label{eq:ObjPDT.3}\\
\text{s.t.} &\quad\eqref{eq:C1PDT},\eqref{eq:C2PDT},\eqref{eq:C5PDT},
\end{align}
\end{subequations}
where $E_{\text{UAV}}^{\{l\}} = T_{s}\sum_{n=1}^{N}P_{\text{UAV}}^{\{l\}}(n)$ is the UAV's total energy consumption, and the constant $c_{2}^{\{l\}}(n)$ is given by
\begin{equation}
  c_{2}^{\{l\}}(n)=\frac{\beta_{0}\,\beta_{m,k}}{\sigma^{2}\left(H^2+\left\|\bw_{m,k}-\bq^{\{l\}}(n)\right\|^2\right)}.
\end{equation}

Note that problem \eqref{eq:PDT.3} is a FP problem, in which the objective \eqref{eq:ObjPDT.3} is a fractional function with a concave numerator and a linear denominator in terms of the CEs' transmission power $\bP$; and the constraints \eqref{eq:C1PDT}, \eqref{eq:C2PDT} and \eqref{eq:C5PDT} are all convex.
Therefore, we can exploit the standard Dinkelbach's method to transform the FP problem \eqref{eq:PDT.3} into its equivalent convex problem \cite{Crouzeix1991}. The optimal solution of problem \eqref{eq:PDT.3} can be obtained by solving the equivalent convex problem iteratively, which can be tackled with existing optimization tools like CVX~\cite{CVXTool2016}.

\vspace{-0.1cm}
\subsection{UAV Trajectory Optimization}
Given the BDs' scheduling matrix $\calB^{\{l\}}$ and the CEs' transmission power $\bP^{\{l\}}$, the UAV's trajectory $\bQ$ can be optimized by solving the following problem
\begin{subequations}\label{eq:PDT.4}
\begin{align}
\underset{\bQ}{\max}  &\quad
\frac{  T_{s} \sum \limits_{n=1}^{N}\sum\limits_{m=1}^{M}\sum\limits_{k=1}^{\barK}b_{m,k}^{\{l\}}(n)\log_{2}\left( 1+\frac{c_{3}^{\{l\}}(n)}{H^2+\left\|\bw_{m,k}-\bq(n)
\right\|^2}\right)}{T_{s}\sum\limits_{n=1}^{N}P_{\text{UAV}}(n)+E_{\text{CE}}^{\{l\}}}\label{eq:ObjPDT.4}\\
\text{s.t.} &\quad\eqref{eq:C1PDT},\eqref{eq:C6PDT},\eqref{eq:C7PDT},
\end{align}
\end{subequations}
where $E_{\text{CE}}^{\{l\}}=T_{s}\sum _{n=1}^{N}\sum _{m=1}^{M}P_{m}^{\{l\}}(n)$ is the CEs' total transmission energy, and the constant $c_{3}^{\{l\}}(n)$ is given by $c_{3}^{\{l\}}(n)=\beta_{0}\,\beta_{m,k}P_{m}^{\{l\}}(n)/\sigma^{2}$.
Obviously, problem \eqref{eq:PDT.4} is a FP problem, which can also be solved iteratively by the Dinkelbach's method. However, the transformed equivalent problem in each iteration is not convex, due to the non-concavity of the logarithmic function with respect to $\{\bq(n)\}$ in the objective \eqref{eq:ObjPDT.4} and constraint \eqref{eq:C1PDT}, as well as the non-convexity of the function $P_{\text{UAV}}(n)$ in the denominator of \eqref{eq:ObjPDT.4}. In the following, the SCA technique is utilized to deal with the non-convexity of the transformed problem.

Preliminarily, similar to \cite{Zeng2018Energy}, we introduce the positive slack variables $\{y(n)\}$'s into the formula of $P_{\text{UAV}}(n)$ 
\begin{equation}\label{eq:slacker}
  y(n)^{2}=\sqrt{1+\frac{V(n)^4}{4v_{0}^4}}-\frac{V(n)^2}{2v_{0}^2}.
\end{equation}
Let $\by=[y(1), \ldots, y(N)]^{\text{T}}$. Substituting \eqref{eq:slacker} into \eqref{eq:PDT.4} yields the following optimization problem
\begin{subequations}\label{eq:PDT.4.1}
\begin{align}
\underset{\bQ, \by}{\max}  &\quad
\frac{  T_{s} \sum \limits_{n=1}^{N}\sum\limits_{m=1}^{M}\sum\limits_{k=1}^{\barK}b_{m,k}^{\{l\}}(n)\log_{2}\left( 1+\frac{c_{3}^{\{l\}}(n)}{H^2+\left\|\bw_{m,k}-\bq(n)
\right\|^2}\right)}{T_s \sum\limits_{n=1}^{N}\left[ P_{\text{b}} \left(1+\frac{3V(n)^2}{U_{\text{tip}}^2}\right)+ P_{\text{i}}y(n)+\frac{d_0\theta sAV(n)^3}{2}\right]+ E_{\text{CE}}^{\{l\}}} \label{eq:ObjPDT.4.1}\\
\text{s.t.} &\quad\frac{1}{y(n)^{2}}\leq y(n)^{2}+\frac{V(n)^{2}}{v_{0}^{2}}, \quad \forall n\label{eq:C1PDT.4.1}\\
&\quad\eqref{eq:C1PDT},\eqref{eq:C6PDT},\eqref{eq:C7PDT}.
\end{align}
\end{subequations}

Noted that problem \eqref{eq:PDT.4.1} is equivalent to problem \eqref{eq:PDT.4}, since the constraint \eqref{eq:C1PDT.4.1} is always satisfied with equality when the optimal solution of this problem is obtained. Otherwise, if a constraint \eqref{eq:C1PDT.4.1} with any $n$ is satisfied with strict inequality, we can always increase the EE by decreasing the value of $y(n)$. However, the FP problem \eqref{eq:PDT.4.1} is still difficult to be solved, due to the non-concave numerator of \eqref{eq:ObjPDT.4.1}, as well as the non-convex constraints \eqref{eq:C1PDT} and \eqref{eq:C1PDT.4.1}.


The basic idea of SCA is to obtain a locally optimal (suboptimal) solution to the original non-convex problem by replacing its non-convex objective and constraints with their upper or lower bounds.
Notice that the logarithmic function in \eqref{eq:ObjPDT.4.1} and \eqref{eq:C1PDT} is convex with respect to the term $H^2+\left\|\bw_{k}-\bq(n)\right\|^2$. From the fact that a convex function is globally lower bounded by its first-order Taylor expansion, we have the following inequality for the logarithmic function
\begin{align}
   \log_{2}&\left( 1  +\frac{c_{3}^{\{l\}}(n)}{H^2+\left\|\bw_{m,k}-\bq(n)\right\|^2}\right) \geq \notag\\
   &\alpha_{m,k}^{\{l\}}(n)-\phi_{m,k}^{\{l\}}(n) \left(\|\bw_{m,k}-\bq(n)\|^2-\|\bw_{m,k}-\bq^{\{l\}}(n)\|^2\right)
   \overset{\Delta}{=} R_{m,k}^{\text{lb}, \{l\}} \left(\bq(n)\right), \label{eq:Taylor1}
\end{align}
where $\alpha_{m,k}^{\{l\}}(\!n\!) \!=\! \log_{2} \!\left(\!\! 1 \!+\! \frac{c_{3}^{\{l\}}(n)}{{H^2 + \left\|\bw_{m,k}-\bq^{\{l\}}(n)\right\|^2}} \! \!\right)$ and
\begin{equation}\label{eq:coef}
   \phi_{m,k}^{\{l\}}(n) = \frac{(\log_{2}e)c_{3}^{\{l\}}(n)}{\left(H^2+\left\|\bw_{m,k}-\bq^{\{l\}}(n)\right\|^2+c_{3}^{\{l\}}(n)\right)\left(H^2+\left\|\bw_{m,k}-\bq^{\{l\}}(n)\right\|^2\right)}.
\end{equation}
Note that the lower bound $R_{m,k}^{\text{lb}, \{l\}}\left(\bq(n)\right)$ in \eqref{eq:Taylor1} is concave with respect to $\bq(n)$.

The right-hand-side (RHS) of constraint \eqref{eq:C1PDT.4.1} is a jointly convex function of $y(n)$ and $V(n)$. By applying the first-order Taylor expansion again, the lower bound for the RHS of \eqref{eq:C1PDT.4.1} can be obtained as
\begin{align}\label{eq:Taylor2}
   \mathstrut &y(n)^{2}+\frac{\left\|\bq(n)-\bq(n-1)\right\|^{2}}{v_{0}^{2}T_{s}^{2}}  \geq y^{\{l\}}(n)^{2}+2y^{\{l\}}(n)
   \left(y(n)-y^{\{l\}}(n)\right) \notag\\
   &\quad-\frac{\left\|\bq^{\{l\}}(n)-\bq^{\{l\}}(n-1)\right\|^{2}}{v_{0}^{2}T_{s}^{2}}+\frac{2}{v_{0}^{2}T_{s}^{2}}
   \left(\bq^{\{l\}}(n)-\bq^{\{l\}}(n-1)\right)^{\text{T}} \left(\bq(n)-\bq(n-1)\right),
\end{align}
where $\{y^{\{l\}}(n)\}$ is the obtained optimal $\{y(n)\}$ in the last iteration. The lower bound in \eqref{eq:Taylor2} is a jointly linear function with respect to $\bq(n)$ and $y(n)$.


Substituting the lower bound $R_{m,k}^{\text{lb}, \{l\}}\left(\bq(n)\right)$ for the logarithmic function in both \eqref{eq:ObjPDT.4.1} and the LHS of \eqref{eq:C1PDT}, and replacing the RHS of \eqref{eq:C1PDT.4.1} with the lower bound in \eqref{eq:Taylor2}, problem \eqref{eq:PDT.4.1} can be reformulated as 
\begin{subequations}\label{eq:PDT.4.2}
\begin{align}
\mspace{-15mu}\underset{\bQ,\by}{\max}  &\quad
\frac{T_s \sum\limits_{n=1}^{N}\sum\limits_{m=1}^{M}\sum\limits_{k=1}^{\barK}b_{m,k}^{\{l\}}(n) R_{m,k}^{\text{lb}, \{l\}}(\bq(n))}{T_s \sum\limits_{n=1}^{N}\left[ P_{0}\left(1+ \frac{3V(n)^2}{U_{\text{tip}}^2}\right)+ P_{i}y(n)+\frac{1}{2}\/d_0\theta sAV\!(n)^3\right]+ E_{\text{CE}}^{\{l\}}}\label{eq:ObjPDT.4.2}\\
\text{s.t.} &\quad T_{s}\sum\limits_{n=1}^{N} b_{m,k}(n)R_{m,k}^{\text{lb}, \{l\}} (\bq(n)) \geq \barQ_{m,k},\quad\forall m,k\label{eq:C1PDT.4.2}\\
&\quad \frac{1}{y(n)^{2}}\leq  y^{\{l\}}(n)^{2}+2y^{\{l\}}(n)\left(y(n)-y^{\{l\}}(n)\right)-
 \frac{\left\|\bq^{\{l\}}(n)-\bq^{\{l\}}(n-1)\right\|^{2}}{v_{0}^{2}T_{s}^{2}} \notag\\
&\quad\quad\quad \quad \ + \frac{2}{v_{0}^{2}T_{s}^{2}}\left(\bq^{\{l\}}(n)-\bq^{\{l\}}(n-1)\right)^{\text{T}} \left(\bq(n)-\bq(n-1)\right),\quad \forall n \label{eq:C2PDT.4.2}\\
&\quad\eqref{eq:C6PDT},\eqref{eq:C7PDT}.
\end{align}
\end{subequations}
The objective \eqref{eq:ObjPDT.4.2} is a fractional function with a concave numerator and a convex denominator, and all the constraints of \eqref{eq:PDT.4.2} are convex. Hence, utilizing the Dinkelbach's method, the optimal solution of problem \eqref{eq:PDT.4.2} can be obtained by solving its equivalent convex problem iteratively.

Note that the achieved maximal EE of problem \eqref{eq:PDT.4.2} is a lower bound of that of problem \eqref{eq:PDT.4.1} which is equivalent to the original problem \eqref{eq:PDT.4}. First, the feasible region in problem \eqref{eq:PDT.4.2} is typically a subset of that in problem \eqref{eq:PDT.4.1}. Once the constraints \eqref{eq:C1PDT.4.2} and \eqref{eq:C2PDT.4.2} are satisfied, the constraints \eqref{eq:C1PDT} and \eqref{eq:C1PDT.4.1} must be satisfied. Second, with the approximation of the logarithmic function in \eqref{eq:Taylor1}, the objective \eqref{eq:ObjPDT.4.2} is a global lower bound of objective \eqref{eq:ObjPDT.4.1}.

\subsection{Overall Algorithm}
In conclusion, the original problem \eqref{eq:PDT} can be efficiently solved through BCD method, when the three subproblems \eqref{eq:PDT.1}, \eqref{eq:PDT.3}, \eqref{eq:PDT.4.2} are alternatively optimized with the local points $\{\calB^{l}, \bP^{l}, \bQ^{l}\}$ updated in each iteration. The overall steps are summarized as Algorithm \ref{AlgorithmP1}.
\begin{algorithm}[t!]
\caption{BCD-based algorithm for solving problem \eqref{eq:PDT}}\label{AlgorithmP1}
\begin{algorithmic}[1]
\STATE Initialize the variables $\calB^{\{0\}}$, $\bP^{\{0\}}$, $\bQ^{\{0\}}$ and positive the threshold $\epsilon$. Let $l=1$.\\
\REPEAT
\STATE Solve \eqref{eq:PDT.1} by the cutting-plane technique for given $\{\bP^{\{l\}}$, $\bQ^{\{l\}}\}$, and obtain the optimal $\calB^{\{l+1\}}$.
\STATE Solve \eqref{eq:PDT.3} by the Dinkelbach's method for given $\{\calB^{\{l+1\}}$, $\bQ^{\{l\}}\}$, and obtain the optimal $\bP^{\{l+1\}}$.
\STATE Solve \eqref{eq:PDT.4.2} by the Dinkelbach's method and the SCA technique for given $\{\calB^{\{l+1\}}$, $\bP^{\{l+1\}}\}$, and obtain the suboptimal $\bQ^{\{l+1\}}$.
\STATE  Update iteration index $l=l+1$.
\UNTIL The increment of the objective function value is smaller than $\epsilon$. 
\STATE Return the optimal solution for \eqref{eq:PDT}, denoted as $\calB^{\star}=\calB^{\{l-1\}}$, $\bP^{\star}=\bP^{\{l-1\}}$, $\bQ^{\star}=\bQ^{\{l-1\}}$.
\end{algorithmic}
\end{algorithm}

\subsection{Convergence and Complexity Analysis}
The classic BCD algorithm converges as long as all subproblems for updating each block of variables are solved optimally in each iteration~\cite{HongLuoBCDSPM17}. However, in the proposed Algorithm \ref{AlgorithmP1}, the approximate problem \eqref{eq:PDT.4.2} of subproblem \eqref{eq:PDT.4} for UAV's trajectory optimization is solved suboptimally. Hence, the convergence analysis for the classic BCD technique cannot be directly used in our case, and the convergence of Algorithm \ref{AlgorithmP1} is proved as follows~\cite{YangYuanLiang19}.
\begin{mythe}
Algorithm \ref{AlgorithmP1} is guaranteed to converge.
\end{mythe}

\begin{proof}
First, in step 3 of Algorithm \ref{AlgorithmP1}, since the optimal solution $\calB^{\{l+1\}}$ of problem \eqref{eq:PDT.1} is obtained optimally for given $\bP^{\{l\}}$ and $\bQ^{\{l\}}$, we have the following inequality
\begin{align}
  {\mathbf{EE}}(\calB^{\{l\}}, \bP^{\{l\}}, \bQ^{\{l\}}) \leq {\mathbf{EE}}(\calB^{\{l+1\}}, \bP^{\{l\}}, \bQ^{\{l\}}). \label{eq:Qinequality1}
\end{align}

%
Second, in step 4 of Algorithm \ref{AlgorithmP1}, since the optimal solution $\bP^{\{l+1\}}$ of problem \eqref{eq:PDT.3} is obtained optimally by using Dinkelbach's method for given $\calB^{\{l+1\}}$ and $\bQ^{\{l\}}$, it holds that
\begin{align}
  {\mathbf{EE}}(\calB^{\{l+1\}}, \bP^{\{l\}}, \bQ^{\{l\}}) \leq  {\mathbf{EE}}(\calB^{\{l+1\}}, \bP^{\{l+1\}}, \bQ^{\{l\}}). \label{eq:Qinequality3}
\end{align}

Third, in step 5 of Algorithm \ref{AlgorithmP1}, it follows that
\begin{align}
 {\mathbf{EE}}(\calB^{\{l+1\}}, \bP^{\{l+1\}}, \bQ^{\{l\}}) &\eqa {\mathbf{EE}}_{\text{lb}} (\calB^{\{l+1\}}, \bP^{\{l+1\}}, \bQ^{\{l\}})\nonumber \\
  &\leb  {\mathbf{EE}}_{\text{lb}} (\calB^{\{l+1\}}, \bP^{\{l+1\}}, \bQ^{\{l+1\}}) \nonumber \\
  &\lec {\mathbf{EE}} (\calB^{\{l+1\}}, \bP^{\{l+1\}}, \bQ^{\{l+1\}}), \label{eq:Qinequality4}
\end{align}
where ($a$) comes from the fact that the Taylor expansion in \eqref{eq:Taylor1} is tight at the given local point $\bQ^{\{l\}}$, i.e., problem \eqref{eq:PDT.4.2} achieves the same maximal EE as problem \eqref{eq:PDT.4} at $\bQ^{\{l\}}$; ($b$) holds since $\bQ^{\{l+1\}}$ is the optimal solution to problem \eqref{eq:PDT.4.2}; and ($c$) holds since the maximal EE of problem \eqref{eq:PDT.4.2} is a lower bound of that of problem \eqref{eq:PDT.4}. The inequality \eqref{eq:Qinequality4} implies that the achieved maximal EE is always non-decreasing after each iteration, although the approximate problem \eqref{eq:PDT.4.2} of the original UAV trajectory optimization subproblem \eqref{eq:PDT.4} is solved locally optimally in each iteration.


From \eqref{eq:Qinequality1}, \eqref{eq:Qinequality3} and \eqref{eq:Qinequality4}, we further have
\begin{align}
{\mathbf{EE}}(\calB^{\{l\}}, \bP^{\{l\}}, \bQ^{\{l\}}) \leq {\mathbf{EE}}(\calB^{\{l+1\}}, \bP^{\{l+1\}}, \bQ^{\{l+1\}}),
\end{align}
which implies that the achieved maximal EE is non-decreasing after each iteration in Algorithm \ref{AlgorithmP1}. Moreover, it can be easily checked that the objective value of problem \eqref{eq:PDT} has some upper bound of finite positive number. As a result, the proposed Algorithm \ref{AlgorithmP1} is guaranteed to converge. This ends the convergence proof. 
\end{proof}


The time complexity of Algorithm \ref{AlgorithmP1} is polynomial. In each iteration, only a linear MIP problem \eqref{eq:PDT.1} is solved by using the cutting-plane technique, and two FP problems \eqref{eq:PDT.3} as well as \eqref{eq:PDT.4.2} are solved by using the Dinkelbach's method which needs to solve a series of convex problems.

\section{Fly-and-Hover Scheme}\label{hover-and-fly}
In this section, we introduce the intuitive fly-and-hover scheme for the UBCN, in which the UAV sequentially hovers at $K$ particular positions with a fixed altitude coordinate $H$ and horizontal coordinates $\tilde{\bq}(i)$ for $i \in \calK$. While the UAV hovers above the horizontal position $\tilde{\bq}(i)$ for a time of $t(i)$ seconds, the UAV collects data from the $i$-th BD with the location coordinate $\bw_i$, $i \in \calK$. Furthermore, the UAV does not communicate with any BD, when it flies from the current hovering point to the next point at a constant speed. After a period of time $T$ seconds, the UAV flies back to its initial point. In general, the fly-and-hover scheme is suboptimal but low-complexity, and thus chosen as a benchmark in this paper.


Note that once the hovering positions and the visiting order are determined, the UAV's flying trajectory is fixed. Given some hovering positions, it is a traveling salesman problem (TSP)~\cite{Graph1976} to find the optimal visiting order for the shortest traveling path. The TSP is NP-hard, and there is no effective method to solve it accurately. When the hovering positions change, the optimal solution to the TSP varies as well, which results in an unaffordable computation complexity. By leveraging classical TSP algorithm~\cite{helsgaun2000}, we obtain the (sub)optimal UAV visiting order for given the BDs' locations, which also determines the BDs' transmission order.

Moreover, to save time for UAV collecting data from ground BDs, we assume that the UAV flies at the maximal speed $V_{\text{max}}$ between any two adjacent hovering positions.
Thus, the UAV's propulsion power in flying status $P_{\text{tra}}$ and the power in hovering status $P_{\text{hov}}$ are all constants, which are given by \eqref{PV} with $V(t)$ replaced by $V_{\text{max}}$ and 0, respectively. Denote the CEs' transmission power vector by $\tilde{\mathbf{p}}=[\tilde{P}(1), \ldots, \tilde{P}(K)]^T  \in \bbR^{K \times 1}$, the UAV hovering time vector by $\bt=[t(1),\ldots, t(K)]^{\text{T}} \in \bbR^{K \times 1}$, and the UAV's trajectory matrix by $\tilde{\bQ} =[\tilde{\bq}(0),\tilde{\bq}(1),\ldots,\tilde{\bq}(K)]^{\text{T}} \in \bbR^{2 \times (K+1)}$, respectively.  The total energy consumption of the UAV can be expressed as
\begin{align}\label{eq:ECEhaf}
  \tilde{E}_{\text{UAV}}(\bt, \tilde{\bQ}) = P_{\text{tra}}\sum\limits_{i=1}^{K}\frac{\left
\|\tilde{\bq}(i)-\tilde{\bq}(i-1)\right\|}{V_{\text{max}}}+P_{\text{hov}}\sum\limits_{i=1}^{K}t(i).
\end{align}

The CE $m \in \calM$ that associates BD $i \in \calK$ transmits carrier signals with power $\tilde{P}(i)$ during the time period when the UAV flies from the hovering point $\tilde{\bq}(i-1)$ to $\tilde{\bq}(i)$ and hovers for $t(i)$ seconds to collect data from BD $i$. The total energy consumption of all CEs is given by
\begin{align}\label{eq:PCEhaf}
  \tilde{E}_{\text{CE}}(\tilde{\bp}, \bt, \tilde{\bQ})= \sum \limits_{i=1}^{K} \tilde{P}(i) \left( t(i)+\frac{\left\|\tilde{\bq}(i)-\tilde{\bq}(i-1)\right\|}{V_{\text{max}}}\negthinspace\right).
\end{align}
For the sake of notational conciseness,
the channel pathloss gain between BD $i$ and its associated CE $m$ is denoted by $\beta_i=\frac{\beta_{0}}{\left\|\bw_{i}-\bu_{m}\right\|^2}$. The harvested energy by BD $i$ can be expressed as
\begin{align}
  \tilde{E}_{\text{BD}}(i)& \eqa \eta_i \beta_i\mspace{-15mu}\sum \limits_{j \in \calK_m \setminus \{i\}} \mspace{-10mu} \tilde{P}(j)  \left( t(j)+\frac{\| \tilde{\bq}(j)-\tilde{\bq}(j-1)\|}{V_{\text{max}}}\!\right) +  \eta_i \beta_i \tilde{P}(i)  \frac{\| \tilde{\bq}(i)-\tilde{\bq}(i-1)\|}{V_{\text{max}}}\nonumber\\
  &=\eta_i\beta_i \sum \limits_{j \in \calK_m} \tilde{P}(j)  \left( t(j)+\frac{\| \tilde{\bq}(j)-\tilde{\bq}(j-1)\|}{V_{\text{max}}}\right) - \eta_i \beta_i \tilde{P}(i)t(i),\label{eq:energy_haf}
\end{align}
where the first term in (a) is the harvested energy of BD $i$ during the period when the UAV flies from hovering positions $\tilde{\bq}(j-1)$ to $\tilde{\bq}(j)$ and serves each BD $j\neq i, j \in \calK_{m}$;
the second term is the harvested energy of BD $i$ during the time when the UAV flies from the position $\tilde{\bq}(i-1)$ to $\tilde{\bq}(i)$. 

The total throughput of all the BDs can be expressed as
\begin{align}\label{eq:throughput_haf}
  \tilde{Q}=\sum\limits_{i=1}^{K} t(i)\log_{2}\left(+\frac{\beta_{0}\beta_{i} \tilde{P}(i) }{\sigma^{2}\left(H^{2}+\left\|\tilde{\bq}(i)-\bw_{i}\right\|^{2}\right)}\right).
\end{align}


Considering \eqref{eq:ECEhaf}, \eqref{eq:PCEhaf}, \eqref{eq:energy_haf} and \eqref{eq:throughput_haf}, the EE maximization problem for the UBCN with the hover-and-fly scheme is formulated as follows
\begin{subequations}\label{eq:PHAF}
\begin{align}
\underset{\tilde{\bp}, \bt, \tilde{\bQ}}{\max}  &\quad
\frac{ \sum\limits_{i=1}^{K}t(i)\log_{2}\left(1+\frac{\beta_{0}\beta_{i}\tilde{P}(i)}{\sigma^{2}\left(H^{2}+\left\|\tilde{\bq}(i)-\bw_{i}\right\|^{2}\right)}\right)}{\tilde{E}_{\text{UAV}}(\bt, \tilde{\bQ})+ \tilde{E}_{\text{CE}}(\tilde{\bp}, \bt, \tilde{\bQ)}}\label{eq:ObjPHAF}\\
\text{s.t.}  &\quad t(i)\log_{2} \left(  1+\frac{\beta_{0}\beta_{i}\tilde{P}(i)}{\sigma^{2}\left(H^{2}+\left\|\tilde{\bq}(i)-\bw_{i}\right\|^{2}\right)}\right)\geq\barQ(i),\quad\forall i \label{eq:C1PHAF} \\
&\quad \eta_i \beta_i \sum \limits_{j \in \calK_m} \tilde{P}(j)  \left( t(j)+\frac{\| \tilde{\bq}(j)-\tilde{\bq}(j-1)\|}{V_{\text{max}}}\right) - \eta_i \beta_i \tilde{P}(i)t(i)  \geq \barE(i),\quad\forall i\label{eq:C2PHAF}\\
&\quad 0\leq \tilde{P}(i)\leq P_{\text{max}},\quad\forall i\label{eq:C4PHAF}\\
&\quad \tilde{\bq}(0) = \tilde{\bq}(K). \label{eq:C5PHAF}\\
&\quad \sum_{i=1}^{K}t(i)+\frac{\left\|\tilde{\bq}(i)-\tilde{\bq}(i-1)\right\|}{V_{\text{max}}} \leq T,\label{eq:C6PHAF}
\end{align}
\end{subequations}
where \eqref{eq:C1PHAF} and \eqref{eq:C2PHAF} are the minimum throughput constraint and the minimum harvested energy constraint for each BD respectively, \eqref{eq:C4PHAF} is each CE's transmission power constraint, \eqref{eq:C5PHAF} means that the UAV finally flies back to its starting point after the time period of $T$ seconds, and \eqref{eq:C6PHAF} is the total operation time constraint for UAV's hovering and flying.



For problem \eqref{eq:PHAF}, the objective \eqref{eq:ObjPHAF} is a fractional function with a non-concave numerator and a non-convex denominator.
The LHSs of the constraints \eqref{eq:C1PHAF} and \eqref{eq:C2PHAF} are all non-concave functions of coupled variables $\tilde{\bp}$, $\bt$ and $\tilde{\bQ}$.
Hence, problem \eqref{eq:PHAF} is a non-convex optimization problem, which is difficult to be solved optimally. Fortunately, there is some resemblance in structure between problem \eqref{eq:PHAF} and problem \eqref{eq:PDT}. We therefore apply a similar strategy for problem \eqref{eq:PHAF} based on the BCD method, where each block of variables are optimized with the other two blocks of variables fixed in each iteration. Thus in the iterative process of the BCD-based algorithm, the original problem \eqref{eq:PHAF} is solved by alternatively optimizing three blocks of variables ($\tilde{\bP},\bt,\tilde{\bQ}$), which corresponds to three subproblems respectively. More detailed analysis is revealed in the following.


In iteration $l \ (l\geq 1)$, given the UAV hovering time $\bt^{\{l\}}$ and the UAV hovering positions $\tilde{\bQ}^{\{l\}}$, the CEs' transmission power optimization problem is a PF problem with a convex feasible region and a fractional objective consisting of a concave numerator and a linear denominator. Hence, this FP problem can be solved optimally by the standard Dinkelbach's method. The details are omitted herein due to space limitations.

Then, given the CEs' transmission power $\tilde{\bp}^{\{l\}}$ and the UAV hovering positions $\tilde{\bQ}^{\{l\}}$, the UAV hovering time optimization is a FP problem with a convex feasible region and a fractional objective consisting of a linear numerator and a linear denominator. This FP problem can also be optimally solved by Dinkelbach's method.

Lastly, given the CEs' transmission power $\tilde{\bp}^{\{l\}}$ and the UAV hovering time $\bt^{\{l\}}$, the optimization problem of UAV hovering positions $\tilde{\bQ}$ can be expressed as
\begin{subequations}\label{eq:P1.4}
\begin{align}
\underset{\tilde{\bQ}}{\max}  &\quad
\frac{ \sum\limits_{i=1}^{K}t^{\{l\}}(i)\log_{2}\left(1+\frac{\kappa^{\{l\}}(i)}{H^{2}+\left\|\tilde{\bq}(i)-\bw_{i}\right\|^{2}}\right)}{ P_{\text{tra}}\sum\limits_{i=0}^{K}\frac{\left
\|\tilde{\bq}(i)-\tilde{\bq}(i-1)\right\|}{V_{\text{mr}}}+E_{\text{hov}}+\tilde{E}_{\text{CE}}(\tilde{\bQ)}}\label{eq:ObjP1.4}\\
\text{s.t.} &\quad \eqref{eq:C1PHAF},\eqref{eq:C2PHAF},\eqref{eq:C5PHAF},\eqref{eq:C6PHAF},\notag
\end{align}
\end{subequations}
where the constant coefficient $\kappa^{\{l\}}(i)=\beta_{0}\beta_{i}\tilde{P}(i)/\sigma^{2}$ and $E_{\text{hov}} = P_{\text{hov}}\sum_{i=1}^{K}t(i)$. For the FP problem \eqref{eq:P1.4}, the Dinkelbach's method cannot be directly used due to the non-concavity of the logarithmic function with respect to $\{\tilde{\bq}(i)\}$ in \eqref{eq:ObjP1.4} and \eqref{eq:C1PHAF}, as well as the non-convex constraint \eqref{eq:C2PHAF}. Again we apply the SCA method to tackle these issues. Similar to problem \eqref{eq:PDT.4.1}, we use the SCA technique to obtain the globally lower bound for the logarithmic function in \eqref{eq:ObjP1.4} and \eqref{eq:C1PHAF} as follows
\begin{align}\label{Taylor3}
   \log_{2}\Big(  1+&\frac{\kappa^{\{l\}}(i)}{H^2+\left\|\bw_{i}-\tilde{\bq}(i)\right\|^2}\Big) \geq \notag \\
   &\tilde{\alpha}^{\{l\}}(i)-\tilde{\phi}^{\{l\}}(i) \left(\|\bw_{i}-\tilde{\bq}(i)\|^2-\|\bw_{i}-\tilde{\bq}^{\{l\}}(i)\|^2\right) \overset{\Delta}{=}\tilde{R}^{\text{lb}, \{l\}} \left(\tilde{\bq}(i)\right),
\end{align}
where $\tilde{\alpha}^{\{l\}}(n) = \log_{2}\left( 1+\frac{\kappa^{\{l\}}(i)}{H^{2}+\left\|\tilde{\bq}^{\{l\}}(i)-\bw_{i}\right\|^{2}}\right)$ and
\begin{equation}\label{log2}
   \tilde{\phi}^{\{l\}}(n) = \frac{(\log_{2}e)\kappa^{\{l\}}(i)}{\left(H^{2}+\left\|\tilde{\bq}^{\{l\}}(i)-\bw_{i}\right\|^{2}+\kappa^{\{l\}}(i)\right)\left(H^{2}+\left\|\tilde{\bq}^{\{l\}}(i)-\bw_{i}\right\|^{2}\right)}.
\end{equation}

And the globally lower bound for the LHS of constraint \eqref{eq:C2PHAF} is given by
\begin{align}\label{Taylor4}
  \left\|\tilde{\bq}(i)-\tilde{\bq}(i-1)\right\|^{2}\geq & -\left\|\tilde{\bq}^{\{l\}}(i)-\tilde{\bq}^{\{l\}}(i-1)\right\|^{2}  \notag\\
  &+2\left(\left\|\tilde{\bq}^{\{l\}}(i)-\tilde{\bq}^{\{l\}}(i-1)\right\|\right)^{\text{T}}(\left\|\tilde{\bq}(i)-\tilde{\bq}(i-1)\right\|).
\end{align}

Substituting the lower bound \eqref{Taylor3} and \eqref{Taylor4} into problem \eqref{eq:P1.4} and introducing slack variables $\{z(i)\}$ to replace $\left\|\tilde{\bq}(i)-\tilde{\bq}(i-1)\right\|$ in constraint \eqref{eq:C2PHAF}, we have the following optimization problem
\begin{subequations}\label{eq:P1.4.1}
\begin{align}
\underset{\tilde{\bQ},\bz}{\max}  &\quad
\frac{ \sum\limits_{i=1}^{K}t^{\{l\}}(i)\tilde{R}^{\text{lb}, \{l\}}(\tilde{\bq}(i))}{ P_{\text{tra}}\sum\limits_{i=0}^{K}\frac{\left
\|\tilde{\bq}(i)-\tilde{\bq}(i-1)\right\|}{V_{\text{mr}}}+E_{\text{hov}}+\tilde{E}_{\text{CE}}(\tilde{\bQ})}\label{eq:ObjP1.4.1}\\
\text{s.t.}& \quad \tilde{R}^{\text{lb}, \{l\}}(\tilde{\bq}(i))\geq\barQ(i),\quad\forall i\label{eq:C1P1.4.1}\\
& \quad\eta_i \beta_i \sum \limits_{j \in \calK_m} \tilde{P}(j) \left( t(j)+\frac{z(j)}{V_{\text{max}}}\right) - \eta_i \beta_i \tilde{P}(i)t(i)  \geq  \barE(i),\quad\forall i\label{eq:C1P1.4.2}\\
& \quad z(i)^{2}\leq -\left\|\tilde{\bq}^{\{l\}}(i)-\tilde{\bq}^{\{l\}}(i-1)\right\|^{2} \notag\\
& \quad\quad\quad\ +2(\|\tilde{\bq}^{\{l\}}(i)-\tilde{\bq}^{\{l\}}(i-1)\|)^{\text{T}}(\left\|\tilde{\bq}(i)-\tilde{\bq}(i-1)\right\|),\quad\forall i\label{eq:C1P1.4.3}\\
& \quad \eqref{eq:C5PHAF},\eqref{eq:C6PHAF}.\notag
\end{align}
\end{subequations}
Problem \eqref{eq:P1.4.1} is a FP problem with a convex feasible region and a fractional objective consisting of a concave numerator and a convex denominator, which can be efficiently solved by the Dinkelbach's method.

In conclusion, the original non-convex problem \eqref{eq:PHAF} for the hover-and-fly scheme can be solved through a BCD-based algorithm, where the local point $\{\tilde{\bp}^{\{l\}}, \bt^{\{l\}}, \tilde{\bQ}^{\{l\}}\}$ is updated by alternately optimizing the three blocks of variables. In addition, the complexity of solving problem \eqref{eq:PHAF} for the hover-and-fly scheme is significantly lower than that for the communicate-while-fly scheme, since the number of optimization variables in problem \eqref{eq:PHAF} are much smaller than that of problem \eqref{eq:PDT}.

\section{Simulation Results} \label{simulation}
In this section, we provide numerical simulation results to verify the performance of the designed UBCN system. We consider a geographical area of size of $56 \times 56~\text{m}^2$ on the ground, where $M=4$ CEs and $K = 12$ BDs are deployed. A equal number of BDs are illuminated by each CE, $\barK = K/M =4$. The frequency for all CEs sending carrier signals is $f_c = 900$~MHz. For the UAV, the flying altitude is $H=20$~m, the maximal flying speed is $V_{\text{max}}=10$ m/s and the receiver noise power is $\sigma^2=-144$ dBm. Table~\ref{table1} lists simulation parameters used in the power-consumption model of rotary-wing UAVs~\cite{Zeng2018Energy}. The flying speed for the minimal propulsion power is calculated from \eqref{PV} as $V_{\text{me}}=5.76$~m/s. The blade profile power $P_{\text{b}} = \delta\rho sA\Omega^{3}R^{3}/8$ and induced power $P_{\text{i}} = (1+k)W^{3/2}/\sqrt{2\rho A}$ are 9.1827~W and 11.5274~W, respectively. Other parameters are set as flowing: the number of time slots is $N = 200$, the maximal CEs' transmission power is $P_{\text{max}}=6$ W, each BD's energy harvesting efficiency is $\eta_{m,k}=0.5$, each BD's requirement of minimal harvested energy is $\barE_{m, k} = 1 \times 10^{-4}$~Joule (J) and each BD's requirement of minimal throughput is $\barQ_{m,k}=\barQ$, $\forall m, k$. In addition, the threshold of algorithm~\ref{AlgorithmP1} is set $\epsilon=10^{-4}$.

\begin{table}[h]
\caption{Parameters of rotary-wing UAV.} \label{table1}
\footnotesize
\centering
\begin{tabular}{c|c|c}
  \hline
  \bfseries{Parameter}&\bfseries{Description}&\bfseries{Value}\\
  \hline
  $W$ & UAV weight in Newton & 4.21 \\
  \hline
  $\rho$ & Air density ($\text{kg/m}^{3}$) & 1.205 \\
  \hline
  $R$ & Rotor radius (m) & 0.3 \\
  \hline
  $A$ & Rotor disc area ($\text{m}^2$) & 0.2827 \\
  \hline
  $\Omega$ & blade angular velocity (r/s) & 200 \\
  \hline
  $U_{\text{tip}}$ & Tip speed of the rotor blade & 60 \\
  \hline
  $b$ & Number of blades & 4 \\
  \hline
  $c$ & Blade or aerofoil chord length & 0.0196 \\
  \hline
  $s$ & Rotor solidity & 0.0832 \\
  \hline
  $S_{\text{FP}}$ & Fuselage equivalent flat plate area ($\text{m}^{2}$) & 0.0118 \\
  \hline
  $d_{0}$ & Fuselage drag ratio & 0.5017 \\
  \hline
  $k$ & Increment correction factor & 0.1 \\
  \hline
  $v_{0}$ & Mean rotor induced velocity in hover & 2.4868 \\
  \hline
  $\delta$ & Profile drag coefficient & 0.012 \\
  \hline
\end{tabular}
\end{table}

\begin{figure}[t!]
\vspace{-0.2cm}
	\centering \includegraphics[width=.7\columnwidth]{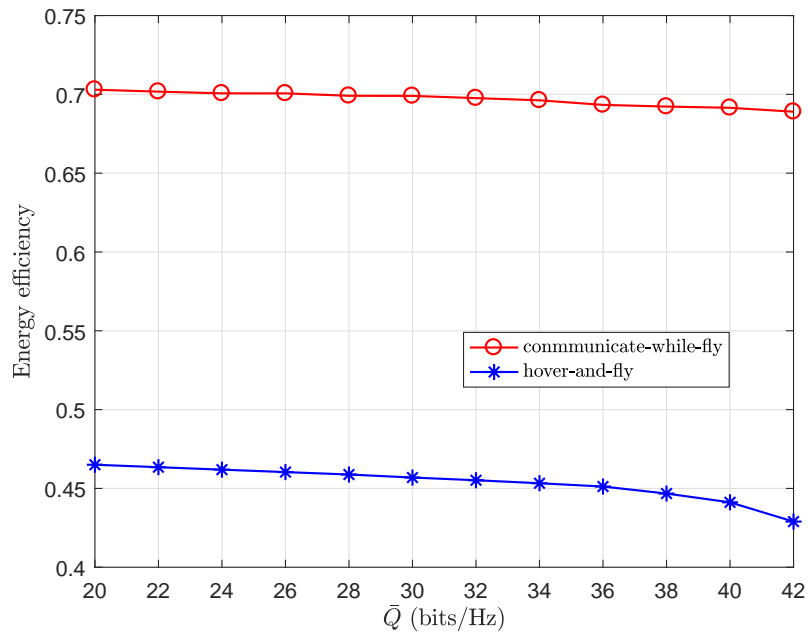}
\caption{Maximal EE versus BD's throughput requirement $\barQ$.} \label{fig:plot_EE_Q}
\vspace{-0.5cm}
\end{figure}

\begin{figure}[htbp]
	\centering \includegraphics[width=.7\columnwidth]{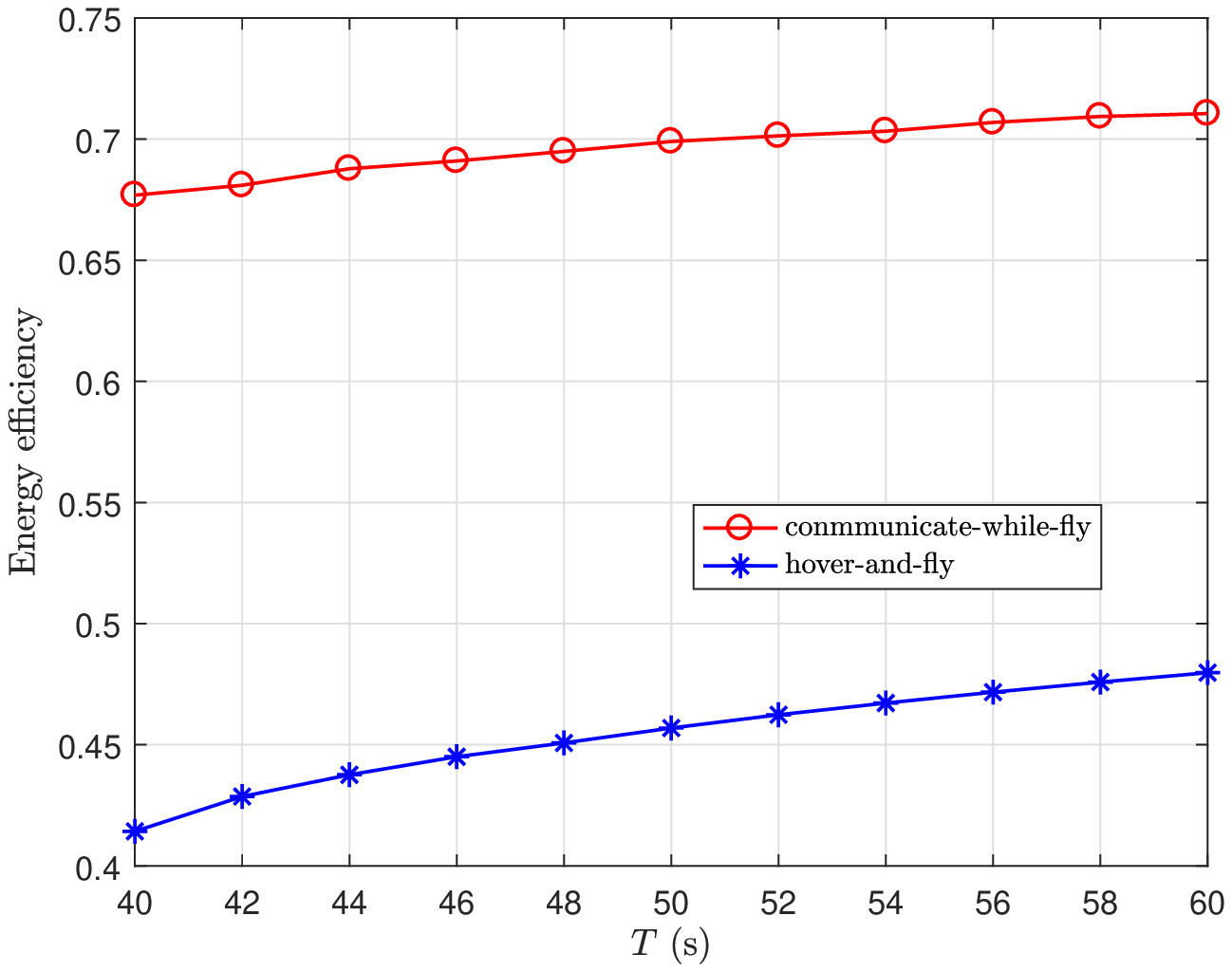}
\caption{Maximal EE versus operation time $T$.} \label{fig:plot_EE_T}
\vspace{-0.5cm}
\end{figure}


Fig.~\ref{fig:plot_EE_Q} plots the maximal EE versus each BD's minimal throughput requirement $\barQ$ for the proposed communicate-while-fly scheme and the benchmark hover-and-fly scheme, with $T = 50$~s. In general, the EE decreases slowly and steadily as $\barQ$ increases from 20~bits/Hz to 42~bits/Hz. Furthermore, the maximal EE obtained in the communicate-while-fly scheme is always significantly higher than that of the hover-and-fly scheme. For instance, the EE achieved by our proposed communicate-while-fly scheme is 53.29\% higher than that achieved by the benchmark scheme, for $\barQ=30$ bits/Hz.

Fig.~\ref{fig:plot_EE_T} plots the maximal EE versus the operation time $T$ for the two schemes, with $\barQ = 30$~bits/Hz. The EE of each scheme increases steadily as $T$ increases. This is because that for longer operation time $T$, the UAV spends a higher proportion of time to perform more energy-efficient data collection via proper trajectory pattern. Moreover, the maximal EE achieved by the communicate-while-fly scheme is significantly higher than that of the hover-and-fly scheme. An approximate increment of 48.1~\% EE gains is achieved by the proposed communicate-while-fly scheme compared with the benchmark, for the case of $T=50~s$.


\begin{figure} [t!]
	\centering \includegraphics[width=.7\columnwidth]{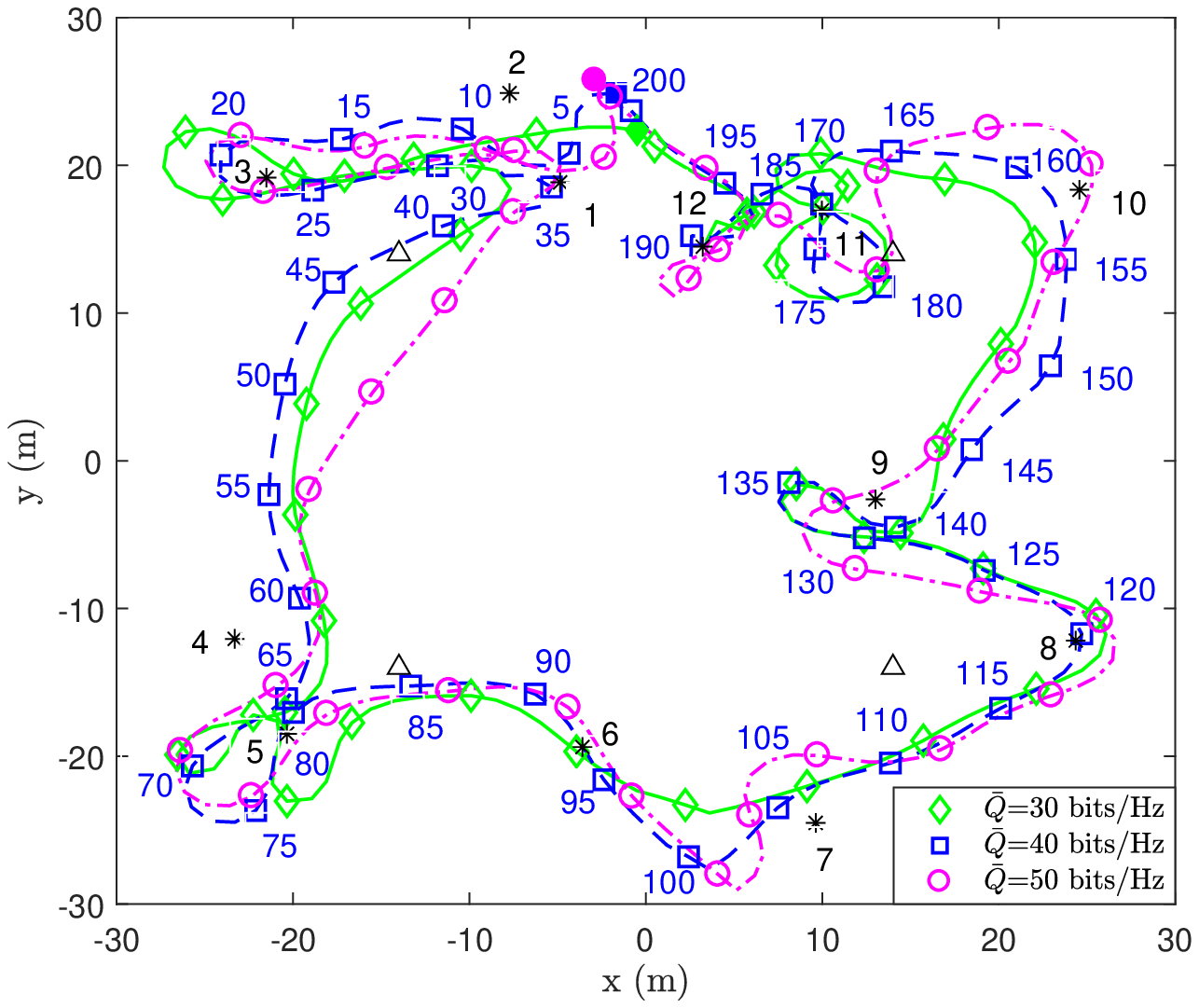}
\caption{Optimal UAV flying trajectory for different $\barQ$. } \label{fig:trajectory_3_Q} 
\end{figure}
\begin{figure} [t!]
	\centering \includegraphics[width=.7\columnwidth]{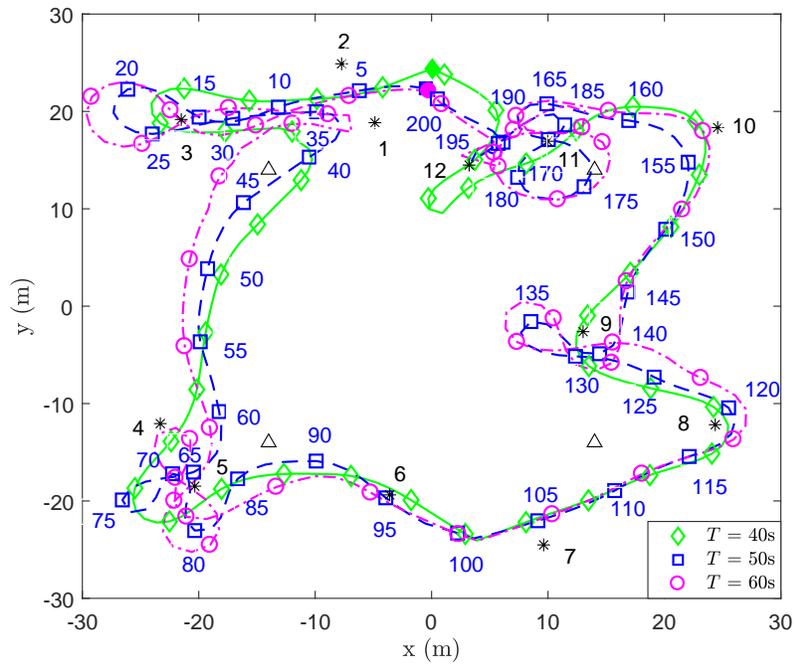}
\caption{Optimal UAV flying trajectory for different $T$. } \label{fig:trajectory_3_T} 
\end{figure}


Fig.~\ref{fig:trajectory_3_Q} and Fig.~\ref{fig:trajectory_3_T} illustrate the UAV's optimal trajectories under different throughput requirements $\barQ$'s with $T = 50$~s and different operation time $T$'s with $\barQ = 30$~bits/Hz, respectively. In the two figures, the locations of CEs and BDs are marked by ``$\bigtriangleup$'' and ``$\ast$'', respectively; and each trajectory curve is sampled every 5 time slots with markers using the same color as the curve. From Fig. \ref{fig:trajectory_3_Q} and Fig.~\ref{fig:trajectory_3_T}, we first observe that the UAV usually hovers around those BDs which are located nearer to their associated CEs (named as nearly-located BDs), such as BD 3, 5, 9 and 11, rather than other BDs.
This is because that the shorter the distance between a BD and its associated CE is, the stronger signals the BD reflects,
and thus easily leading to a higher data rate (and EE).
Furthermore, the UAV's trajectory has the ``8''-shape, when it collects data from those nearly-located BDs like BD 5 and 11. This interesting phenomenon is typically an energy-efficient flying pattern of fixed-wing UAVs appeared in~\cite{ZengZhangEETWC17}. However, for higher throughput requirement $\barQ$ of each BD (e.g., 50~bits/Hz) or less available time resource $T$ (e.g., 40~s), the ``8''-shape trajectory becomes less obvious such that more stringent throughput and time constraints can be satisfied.

\begin{figure} [t!]
	\centering \includegraphics[width=.7\columnwidth]{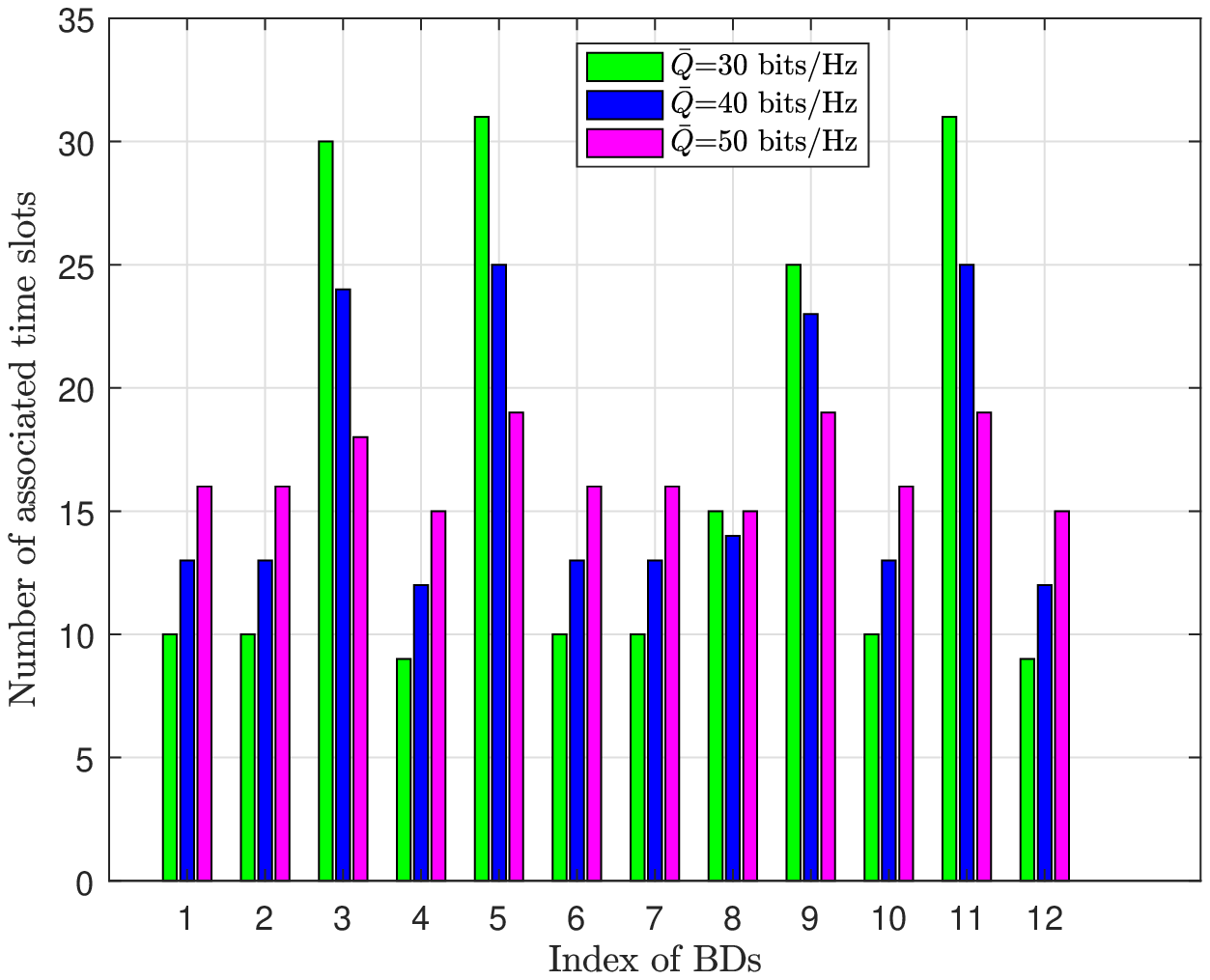}
\caption{Optimal time resource allocation among BDs for different $\barQ$.}\label{fig:TimeAllocation_3_Q}
\vspace{0.1cm}
\end{figure}
\begin{figure} [t!]
	\centering \includegraphics[width=.7\columnwidth]{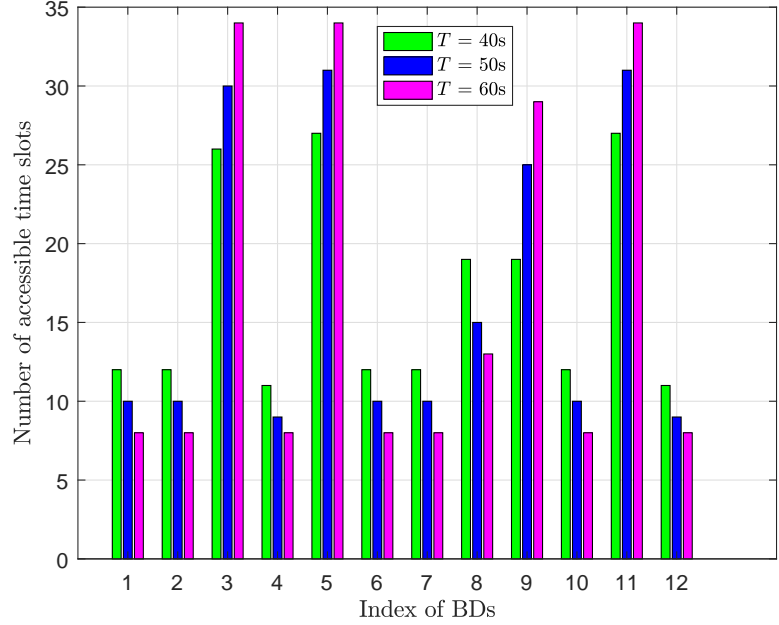}
\caption{Optimal time resource allocation among BDs for different $T$.}\label{fig:TimeAllocation_3_T}
\vspace{0.1cm}
\end{figure}


Fig.~\ref{fig:TimeAllocation_3_Q} and Fig.~\ref{fig:TimeAllocation_3_T} illustrate the time resource allocation among BDs, under the same parameter settings as Fig.~\ref{fig:trajectory_3_Q} and Fig.~\ref{fig:trajectory_3_T}, respectively. In the two figures, the nearly-located BDs are usually allocated with more transmission slots, such as BD 3, 5, 9 and 11, since a higher data rate (and EE) can be achieved when the UAV flies around the nearly-located BDs with the same UAV-power consumption. This observation coincides with the phenomenon that the UAV hovers around the nearly-located BDs, as illustrated in Fig.~\ref{fig:trajectory_3_Q} and Fig.~\ref{fig:trajectory_3_T}. Also, we note that the number of time slots allocated for each BD becomes more balanced, with the increase of $\barQ$ or the decrease of $T$. This is because that sufficient time has to be scheduled to each BD with longer CE-to-BD distance to satisfy its throughput requirement $\barQ$, when $\barQ$ increases or the total operation time $T$ decreases.



\begin{figure} [t!]
	\centering \includegraphics[width=.7\columnwidth]{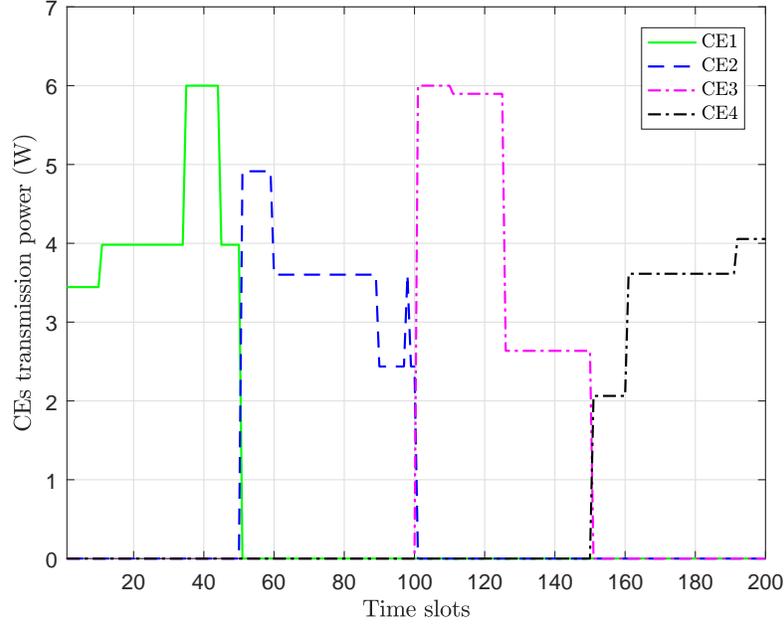}
\caption{CEs' transmission power versus time slot index.} \label{fig:CE_transmission_power}
\end{figure}

Fig.~\ref{fig:CE_transmission_power} plots each CE's transmission power in the period of the operation time $T = 50$~s, with $\barQ = 30$~bits/Hz. In any time slot, only one CE is turned on to illuminate all the BDs within its coverage. The CEs' emitting order is also optimized to conform to the UAV's trajectory design and the BDs' time resource allocation. To maximize the UBCN's EE, the optimal transmission power of each illuminating CE varies from around $2$~W to $6$~W, which keeps a relatively lower value most of the time in the whole operation period, rather than the maximal transmission power $8$~W.

\begin{figure} [t!]
	\centering \includegraphics[width=.7\columnwidth]{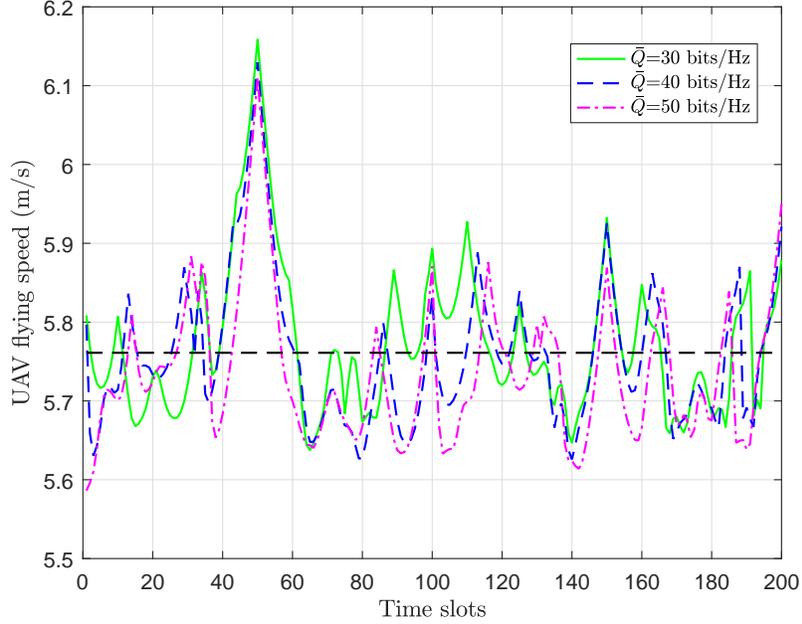}
\caption{Optimal UAV flying speed versus time slot index.} \label{fig:speed_3_Q}
\end{figure}
Fig.~\ref{fig:speed_3_Q} plots the UAV's flying speed versus the time slots, with $T = 50$~s and $\barQ = 30, 40, 50$~bits/Hz. In general, the fluctuation of the UAV's flying speed is relatively small during the whole operation period of $T$ seconds, under different $\barQ$'s. We further observe that when the nearly-located BDs communicate with the UAV, the UAV's flying speed fluctuates around 5.76~m/s, depicted by the black dash line in Fig.~\ref{fig:speed_3_Q}, which is exactly the speed $V_{\text{me}}$ leading to lowest propulsion power for the UAV. Moreover, the UAV neither remain static nor flies at the maximal velocity, because these two extreme cases result into higher propulsion power consumption calculated in \eqref{PV}.

\begin{figure} [t!]
	\centering \includegraphics[width=.7\columnwidth]{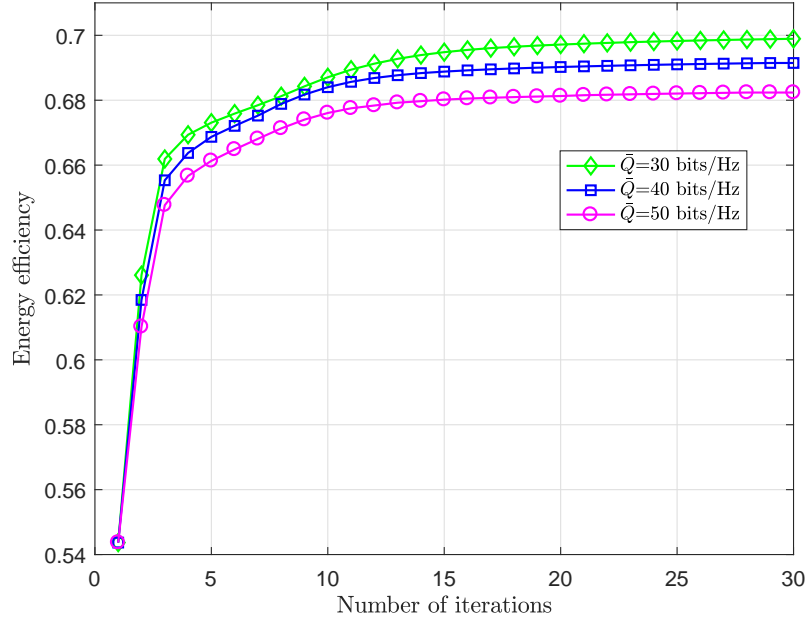}
\caption{Convergence of Algorithm \ref{AlgorithmP1} for energy efficiency minimization.} \label{fig:convergence_3_Q}
\end{figure}

Finally, Fig.~\ref{fig:convergence_3_Q} illustrates the convergence performance of the proposed Algorithm~\ref{AlgorithmP1}, with $T = 50~s$ and $\barQ = 30, 40, 60$~bits/Hz, respectively. We observe that the proposed algorithm converges within less than 25 iterations, and the EE increases significantly in the first about 5 iterations. The fast convergence of the proposed algorithm is validated. Moreover, the optimal EE for $\barQ = 30$~bits/Hz, $\barQ = 40$~bits/Hz and $\barQ = 60$~bits/Hz finally converges to 0.699, 0.691 and 0.682.

\vspace{6cm}
\section{Conclusion}\label{conslusion}
This paper has studied a UAV-assisted backscatter communication network, where the energy efficiency (EE) is maximized by jointly optimizing the BDs' scheduling, the CEs' transmission power, and the UAV's trajectory. The formulated EE optimization problem is solved by the proposed BCD-based iterative algorithm which also utilizes the cutting-plane method, the Dinkelbach's method and the SCA technique. Numerical results show that the proposed communicate-while-fly scheme achieves significant EE gains compared with the benchmark hover-and-fly scheme, and also validate the convergence of the proposed algorithm. Some useful insights on energy-efficient UAV-trajectory design and resource allocation are also obtained. In order to achieve the maximal EE, the UAV is optimized to hover around the nearly-located BDs at the specific speed with minimal propulsion energy consumption, and more time resource is allocated to the nearly-located BDs. There are some interesting future work such as the extension to the optimization of three-dimensional trajectory, the case of multi-antenna UAV, and other multiple-access schemes (e.g., NOMA).

\bibliography{IEEEabrv,reference20191102}
\bibliographystyle{IEEEtran}

\end{document}